%% file: main.tex
\documentclass{article}

\usepackage{microtype}
\usepackage{graphicx}
\usepackage{subfigure}
\usepackage{amsfonts}
\usepackage{booktabs} 
\usepackage{multirow}

\usepackage{amsmath}
\usepackage{hyperref}
\usepackage{cleveref}
\usepackage{nicefrac}
\usepackage{amssymb}
\usepackage[symbol]{footmisc}

\crefname{algocf}{alg.}{algs.}
\Crefname{algocf}{Algorithm}{Algorithms}

\usepackage[accepted]{icml2021}

\usepackage{amsthm}
\usepackage{mathtools}
\usepackage{relsize}
\usepackage{bbm}

\newcommand*{\defeq}{\stackrel{\mathsmaller{\mathsf{def}}}{=}}
\newcommand*{\disteq}{\stackrel{d}{=}}
\newcommand{\rmax}{r_{max}}

\newcommand{\Nats}{\mathbb{N}}
\newcommand{\Reals}{\mathbb{R}}

\newcommand{\parent}{\mathrm{par}}

\newcommand{\Prob}{\mathbb{P}}
\newcommand{\Exp}{\mathbb{E}}

\newcommand\eqdef{\stackrel{\tiny \text{def}}{=}}
\newcommand\eqdist{\stackrel{\tiny \text{d}}{=}}

\newtheorem{theorem}{Theorem}
\newtheorem{lemma}{Lemma}
\newtheorem{definition}{Definition}
\newtheorem{assumption}{Assumption}
\newtheorem{corollary}{Corollary}

\newtheorem*{theorem*}{Theorem}

\input{math_commands.tex}

\newcommand{\GumbelProcess}{\mathcal{GP}}
\newcommand{\TruncGumbel}[2]{\mathrm{TG}\left(#1, #2\right)}
\newcommand{\Uniform}[2]{\mathcal{U}\left(#1, #2\right)}
\newcommand{\ApproxExp}{\mathcal{I}}
\newcommand{\Oh}{\mathcal{O}}
\newcommand{\partition}{\texttt{partition} }
\newcommand{\Ancestors}{\mathcal{A}}
\newcommand{\Frontier}{\mathcal{F}}
\newcommand{\Tree}{\mathcal{T}}
\newcommand{\Ind}{\mathbbm{1}}
\renewcommand{\Var}{\mathbb{V}}

\newcommand{\Dinf}[2]{D_{\infty}[#1||#2]}

\newcommand{\nil}{\mathtt{nil}}

\DeclarePairedDelimiterX{\infdivx}[2]{[}{]}{%
  #1\delimsize\| #2%
}
\newcommand{\KLD}{\KL\infdivx}
\newcommand{\infD}{D_{\infty}\infdivx}
\DeclarePairedDelimiter{\norm}{\lVert}{\rVert}
\DeclarePairedDelimiter{\abs}{\lvert}{\rvert}
\DeclarePairedDelimiterX{\innerProd}[2]{\langle}{\rangle}{%
    #1,#2%
}

\icmltitlerunning{Fast Relative Entropy Coding with A* Coding}

\begin{document}

\twocolumn[
\icmltitle{Fast Relative Entropy Coding with A* coding}




\begin{icmlauthorlist}
\icmlauthor{Gergely Flamich$^*$}{cam}
\icmlauthor{Stratis Markou$^*$}{cam}
\icmlauthor{Jos\'e Miguel Hern\'andez-Lobato}{cam,msr,tur}
\end{icmlauthorlist}

\icmlaffiliation{cam}{Department of Engineering, University of Cambridge, Cambridge, UK}
\icmlaffiliation{msr}{Microsoft Research, Cambridge, UK}
\icmlaffiliation{tur}{Alan Turing Institute, London, UK}

\icmlcorrespondingauthor{Gergely Flamich}{gf332@cam.ac.uk}
\icmlcorrespondingauthor{Stratis Markou}{em626@cam.ac.uk}

\icmlkeywords{Machine Learning, ICML}

\vskip 0.3in
]



\printAffiliationsAndNotice{$^*$Equal contribution.} 

\begin{abstract}
Relative entropy coding (REC) algorithms encode a sample from a target distribution $Q$ using a proposal distribution $P$, such that the expected codelength is $\Oh(\KLD{Q}{P})$.
REC can be seamlessly integrated with existing learned compression models since, unlike entropy coding, it does not assume discrete $Q$ or $P$, and does not require quantisation.
However, general REC algorithms require an intractable $\Omega(e^{\KLD{Q}{P}})$ runtime.
We introduce AS* and AD* coding, two REC algorithms based on A* sampling.
We prove that, for continuous distributions over $\mathbb{R}$, if the density ratio is unimodal, AS* has $\Oh(\infD{Q}{P})$ expected runtime, where $\infD{Q}{P}$ is the R\'enyi $\infty$-divergence.
We provide experimental evidence that AD* also has $\Oh(\infD{Q}{P})$ expected runtime.
We prove that AS* and AD* achieve an expected codelength of $\Oh(\KLD{Q}{P})$.
Further, we introduce DAD*, an approximate algorithm based on AD* which retains its favourable runtime and has bias similar to that of alternative methods.
Focusing on VAEs, we propose the IsoKL VAE (IKVAE), which can be used with DAD* to further improve compression efficiency.
We evaluate A* coding with (IK)VAEs on MNIST, showing that it can losslessly compress images near the theoretically optimal limit.
\end{abstract} 

\section{Introduction}
\par
In recent years, there has been significant progress in compression using machine learning, an approach known as learned compression.
Most of the prominent learned compression methods, including the state-of-the-art in both lossless \citep{townsend2019hilloc, hoogeboom2019integer, zhang2021iflow} and lossy compression \citep{balle2017end}, perform non-linear transform coding  \citep{balle2020nonlinear}.

\par
In transform coding, a datum is first mapped to a latent representation and encoded with entropy coding. 
Entropy coding assumes that the latent representation and the coding distribution are discrete, which requires a non-differentiable quantization step.
Since gradient-based optimization requires derivatives, most state-of-the-art methods use a continuous approximation to quantization during training and switch to hard quantization only during compression time \cite{balle2017end}.
This mismatch has been argued to be harmful towards the compression efficiency of these methods \citep{havasi2018minimal,flamich2020compressing, Theis2021a}.

\par
Relative entropy coding \citep[REC;][]{flamich2020compressing} is a recently proposed alternative, which does not require quantization and avoids this mismatch.
A REC algorithm uses samples from a proposal distribution $P$ to produce a random code representing a sample from a target distribution $Q$, with expected length of approximately $\KLD{Q}{P}_2$,
where the subscript specifies that the KL is measured in bits rather than nats.
General-purpose REC algorithms place no restrictions on $Q$ and $P$ beyond that $\KLD{Q}{P}_2$ be finite, so they can be applied even when $Q$ and $P$ are continuous.
Thus REC can be naturally applied to perform compression with generative models trained via gradient descent, for applications including but not limited to: (1) data compression with variational autoencoders \citep[VAE;][]{kingma2013auto}, where $Q$ corresponds to a variational posterior over latent variables and $P$ to a prior over these latent variables; (2) model compression \cite{havasi2018minimal}, where $Q$ corresponds to an approximate posterior over parameters \cite{pmlr-v37-blundell15}, and $P$ to a prior over those parameters.

\par
However, REC algorithms that make no further assumptions on $Q$ or $P$ require $\Omega\left(2^{\KLD{Q}{P}_2} \right)$ steps to terminate in expectation \citep{agustsson2020universally}, which is a severe limitation in practice.
Thus, in order to make progress, it is necessary to impose additional assumptions on $Q$ and $P$.
Universal quantisation \citep{ziv1985universal} can be regarded as a REC algorithm that achieves $\Oh(\KLD{Q}{P}_2)$ runtime and has been demonstrated to work well with the  state-of-the-art VAE-based learned compression methods \citep{agustsson2020universally}. However, it places heavy limitations on $Q$ and $P$, which might be overly restrictive in many cases.

\par
In this work, we introduce AS* and AD* coding, two closely related REC algorithms based on A* sampling \citep{maddison2014sampling}, which achieve significantly faster runtimes than existing alternatives.
For $Q$ and $P$ over $\mathbb{R}$, and without further assumptions, we show that the expected codelength achieved by AS* and AD* is bounded by
\vspace{-0.07cm}
\begin{equation} \label{eq:intro_codelength}
    \lambda~\KLD{Q}{P}_2 + \lambda~\log_2 (\KLD{Q}{P}_2 + 1) + \Oh(1),\vspace{-0.05cm}
\end{equation}
where $\lambda \approx 2.41$ and $\lambda = 1$ for AS* and AD* respectively.
With the additional assumption that $Q$ and $P$ are continuous with unimodal density ratio, we show that the expected runtime of AS* is $\Oh(\infD{Q}{P})$, where $\infD{Q}{P} = \log \sup_{x \in \mathcal{X}}\frac{dQ}{dP}(x)$ is the R\'enyi $\infty$-divergence.
While we do not prove an analogous bound for the runtime of AD*, we conduct extensive experiments on different $Q, P$ pairs, and observe that the runtime of AD* is also $\Oh(\infD{Q}{P})$.
Thus AS* and AD* significantly improve upon the runtime of existing REC methods, without requiring as severe assumptions as universal quantization.
While AS* and AD* require unimodality, this assumption is satisfied by many models in learnt compression, such as most VAEs.

\par
In addition, a practical limitation of REC algorithms is that, since the codelength of a sample is random, additional bits must be used to communicate the codelength itself to ensure the message is decodable.
This additional code, which corresponds to the second and third terms of \cref{eq:intro_codelength}, accounts for a large portion of the overall codelength and grows linearly with the number of dimensions.
To remedy this, we consider approximate REC algorithms, in which the codelength is a parameter, which is set prior to coding.
This allows us to form blocks of variables which are coded using the same codelength.
Thus, the additional cost must be paid only once per block, rather than once per variable,
thereby greatly reducing this overhead codelength.

\par
To this end, we first introduce an approximate variant of AD* coding which, similarly to existing approximate REC algorithms, has a tunable codelength and a provably low bias.
Unlike existing methods however, it retains the favourable runtime of AD* coding.
Second, we propose to parameterize latent variable distributions by explicitly specifying $\KLD{Q}{P}$.
For example, instead of parameterizing a Gaussian $Q$ using a mean and variance, we can specify its mean and $\KLD{Q}{P}$, from which the variance is uniquely determined.
This allows us to construct blocks of latent variables with tied KL divergences, which can be coded with the same codelength.
This codelength must be communicated once per block instead of once per variable.
We consider VAE models using this parameterization, which we refer to as isoKL VAEs (IKVAEs).
We present experiments on lossless image compression on MNIST which demonstrate that the performance of IKVAEs is comparable to that of VAEs, while reducing the codelength overhead.

Our contributions can be summarised as follows:
\vspace{-0.35cm}
\begin{itemize}
    \item We introduce AS* and AD* coding, two REC algorithms based on A* sampling, for coding samples from one-dimensional distributions. \vspace{-0.15cm}
    \item We prove that, the expected codelength of AS* and AD* is $\Oh(\KLD{Q}{P})$. \vspace{-0.15cm}
    \item  We prove that if $dQ/dP$ is bounded and unimodal, AS* achieves $\mathcal{O}(\infD{Q}{P})$ runtime.
    We present empirical evidence that the runtime of AD* is also linear in $D_{\infty}$.
    Therefore AS* and AD* significantly improve over the exponential runtime of existing REC algorithms.
    A direct consequence of the above is that A* sampling with unimodal $dQ/dP$ also has $\mathcal{O}(\infD{Q}{P})$ runtime. \vspace{-0.15cm}
    \item We introduce an approximate variant of AD* and bound its bias.
    Similar to existing ones, this algorithm can code low-bias samples using fixed codelengths, but retains the favourable runtime of AD*. \vspace{-0.15cm}
    \item We introduce a novel modification for VAEs, in which the KL divergences across either all, or some of, the latent dimensions are tied.
    This modification, which we refer to as the isoKL VAE (IKVAE), can be used with any fixed-codelength approximate REC algorithm, such as our own, to greatly reduce overhead codes. \vspace{-0.15cm}
    \item We demonstrate the favourable performance of AS* and AD* on toy problems, comparing it with that of alternative methods.
    Lastly, we apply our approximate AD* algorithm to VAE and IKVAE models on image data, showing that it can losslessly compress images near the theoretically optimal ELBO.
\end{itemize}
\vspace{-0.3cm}
\section{Background}
\vspace{-0.1cm}
\textbf{Relative Entropy Coding:} The central problem which we tackle in this work is the REC problem, defined as follows.
\begin{definition}[REC problem and algorithm]
\label{def:rec_definition}
    Let $Q$ be a target and $P$ be a proposal distribution, with $\KLD{Q}{P} < \infty$, and let $S = (s_1, s_2, \dots)$ be an infinite sequence of publicly available independent fair coin tosses.
    Relative entropy coding (REC) is the problem of producing a uniquely decodable code $C$ representing a sample from $Q$ given $S$, such that the codelength $\abs{C}$ satisfies
    \begin{equation}
    \label{eq:rec_codelength_constraint}
        \mathbb{E}[\abs{C}] = \Oh(\KLD{Q}{P}),
    \end{equation}
    An algorithm which solves this problem is a REC algorithm.
\label{def:rec}
\end{definition}
\vspace{-0.5cm}
In practice, $S$ is implemented by using a pseudo-random number generator (PRNG) with a publicly available seed.
Crucially, REC applies to both discrete and continuous distributions, and can be integrated into learned compression pipelines, without requiring quantization.
Several existing algorithms solve the REC problem without further assumptions on $Q$ or $P$, however, they are impractically slow.

\textbf{Poisson Functional Representation:} \citet{li2018strong} introduced a REC algorithm for general $Q$ and $P$, here referred to as Poisson functional representation (PFR) coding.
\citeauthor{li2018strong} showed that if $T_i$ are the ordered arrival times of a homogeneous Poisson process on $\mathbb{R}^+$ \cite{kingman1992poisson}, and $X_i \sim P$, then
\begin{equation} \label{eq:pfr}
    \argmin_{i \in \mathbb{N}} \left\{ T_i \cdot \frac{dP}{dQ}(X_i)  \right\}\sim Q.
\end{equation}
Further, \citeauthor{li2018strong} showed that a sample may be represented by coding the index which minimises \cref{eq:pfr}, and bounded the expected codelength of PFR by
\begin{equation} \label{eq:pfr_codelength}
    K \leq \mathbb{E}[\abs{C}] \leq K +  \log_2(K + 1) + \mathcal{O}(1),
\end{equation}
where $K = \KLD{Q}{P}_2$.
We note that PFR converts REC into a search problem, just like the A* sampling algorithm \cite{maddison2014sampling} converts sampling into a search problem.
In fact, it can be shown that the minimization in \cref{eq:pfr}, is equivalent to a variant of A* sampling, called Global Bound A* sampling.
In particular
\begin{align} \label{eq:prfastar}
    \argmin_{i \in \mathbb{N}}\left\{T_i \frac{dP}{dQ}(X_i) \right\} &= \argmax_{i \in \mathbb{N}} \left\{ -\log T_i + \log r(X_i) \right\}\nonumber \\
    &\disteq \argmax_{i \in \mathbb{N}} \left\{ G_i + \log r(X_i) \right\}\hspace{-10pt}
\end{align}
where $r = dQ/dP$, and $G_i$ is sampled according to
\begin{equation}
G_i \sim \TruncGumbel{0}{G_{i-1}},
\end{equation}
where $\TruncGumbel{\mu}{\kappa}$ denotes the Gumbel distribution with mean $\mu$ and unit scale, truncated to the interval $(-\infty, \kappa]$, and defining $G_0 \defeq \infty$.
The maximisation in \cref{eq:prfastar} is identical to the Global Bound A* sampling algorithm (see Appendix of \citeauthor{maddison2014sampling}).
Thus, Global Bound A* sampling and PFR are identical, with the exception that the former works in the negative log-space of the latter (eq. \ref{eq:prfastar}).

\par
Unfortunately, the runtime $T$ of PFR is so large that it renders the algorithm intractable in practice.
In particular, the runtime of Global Bound A*, and thus also of PFR, can be shown (see Appendix in \citeauthor{maddison2014sampling}) to be equal to
\begin{equation} \label{eq:pfrruntime}
    \mathbb{E}[T] = \exp\left(\infD{Q}{P}\right) \geq \exp\left(\KLD{Q}{P}\right).
\end{equation}
where $\infD{Q}{P} = \log \sup_{x \in \mathcal{X}}r(x)$.
This bound is perhaps unsurprising, considering a more general result by \citet{agustsson2020universally}, who proved that without further assumptions on $Q$ and $P$, the expected runtime of any REC algorithm is $\Omega(2^{\KLD{Q}{P}_2})$.

\par
\textbf{Additional assumptions:} In order to develop a REC algorithm that is fast enough to be practical, we must make further assumptions about the target and proposal distributions.
Focusing on continuous distributions over $\mathbb{R}$, we will show that an assumption which enables fast runtimes is the unimodality of $r$. While somewhat restrictive, this assumption is satisfied by virtually all models used in learned compression.
We will show that A* sampling can be modified to solve the REC problem, achieving $\Oh(\infD{Q}{P})$ runtime whenever $r$ is bounded and unimodal.

\par
Henceforth, we will assume that $Q$ and $P$ are continuous distributions on $\mathbb{R}$ with densities $q$ and $p$.
However, we note that the methods we present in this work can be generalised to arbitrary measure spaces equipped with a total ordering over their elements, using the Radon-Nikodym derivative in place of the density ratio \cite{grimmett2001probability}.

\par
\textbf{A* sampling:} The A* sampling algorithm \cite{maddison2014sampling} is an extension of the Gumbel-max trick \cite{papandreou2011perturb} to arbitrary probability spaces.
A* sampling is a branch-and-bound algorithm \cite{land1960automatic}, which converts the problem of sampling from $Q$ into the maximization in \cref{eq:prfastar}.
A* sampling builds a binary search tree, where each node $n$ is associated with a triplet $(B_n, X_n, G_n)$, where: (1) $B_n \subseteq \Omega$ is a subset of the sample space;
(2) $X_n \sim P\lvert_{B_n}$ is a sample distributed according to the restriction of $P$ to $B_n$; (3) $G_n \sim \TruncGumbel{\log P(B_n)}{G_{\parent(n)}}$ is a truncated Gumbel sample, where $\parent(n)$ is the parent of $n$.
At each step, A* sampling expands the children of $n$, partitioning $B_n$ into two disjoint subsets $L_n \cup R_n = B_n$, using some rule which we refer to as \texttt{partition}.
A* sampling then constructs the children's triplets following the definitions above.
It then bounds the objective from \cref{eq:prfastar} on each of the children, and utilises the bounds to narrow the search and quickly locate the maximum.
Setting \partition to be the degenerate splitting
\begin{equation} \label{eq:gbap}
    \partition_{GBA^*}(B, X) \defeq (\emptyset, B)
\end{equation}
yields the Global Bound A* algorithm.
\Cref{eq:gbap} hints at why Global Bound A*, and by extension PFR, have large runtimes.
This \partition function does not refine the search as the algorithm progresses, maintaining the same fixed global bound throughout a run.
In this work we consider using two different \partition functions, yielding the AS* and AD* coding algorithms.
When applied to unimodal $r$, these partitioning schemes enable the algorithm to quickly refine its search and achieve a fast runtime.

\par
\textbf{Approximate REC algorithms:} Other lines of work have introduced alternative algorithms, such as Minimal Random Coding \citep[MRC;][]{havasi2018minimal} and Ordered Random Coding \citep[ORC;][]{theis2021algorithms}, which produce a code representing an approximate, instead of an exact, sample from $Q$.
Because these algorithms produce biased samples, strictly speaking they do not satisfy the requirements of \cref{def:rec}, so we refer to them as approximate REC algorithms.
Both MRC and ORC accept the codelength $|C|$ as an input parameter.
Increasing $\abs{C}$ reduces sample bias, but increases compression cost and runtime.
Unfortunately, the runtime required to reduce the bias sufficiently in order to make the samples useful in practice scales as $\mathcal{O}(2^{\KLD{q}{p}_2})$, making MRC and ORC as expensive as PFR.

\par
However, one benefit of a tunable codelength is that, when communicating multiple samples, the overhead code corresponding to the second and third terms in \cref{eq:pfr_codelength}, can be significantly reduced.
By grouping several random variables into blocks, and coding all samples of a block with the same codelength, we only need to communicate a codelength per block, as opposed to one codelength per variable.
This procedure reduces the codelength overhead by a factor equal to the number of variables in each block.
\vspace{-0.3cm}
\section{A* Coding}
\vspace{-0.1cm}
\par
A* sampling returns a node $n$ with associated triplet $(B_n, X_n, G_n)$, where $X_n$ is an exact sample from $Q$ \cite{maddison2014sampling}.
Therefore, the only addition we need to make to A* sampling to turn it into a REC algorithm, is a way to represent $n$ using a uniquely decodable code $C$.
Given such a $C$, we can decode the sample $X_n$ by determining the $B_n$ corresponding to $n$, and then sampling $X_n \sim P|_{B_n}$ given the public source of randomness $S$.

\par
Since A* sampling may return any node in its binary search tree, we propose to use heap indexing, also known as the \textit{ahnentafel} or \textit{Eytzinger ordering} to code nodes. 
Let the parent of the root node be $\nil$.
Then, the heap index of a node $n$ is
\begin{equation}
    H_n \defeq 
    \begin{cases}
    1 & \text{if } \parent(n) = \nil \\
    2H_{\parent(n)} & \text{if } n \text{ is left child of } \parent(n) \\
    2H_{\parent(n)} + 1 & \text{if } n \text{ is right child of } \parent(n).
    \end{cases}
\end{equation}
Let $D_n$ denote the depth of node $n$ in the binary tree, where $D_{\texttt{root}} = 1$.
We can see that $D_n = \lfloor \log_2 H_n \rfloor + 1$.
Thus, assuming that $D_n$ is known, $H_n$ can be encoded in $D_n$ bits.
Therefore, we can modify A* sampling to return the heap index of the optimal node, from which $n$ can be decoded.
This yields \Cref{alg:a_star_coding}, which we refer to as A* coding.
\citet{maddison2014sampling} show that A* sampling is correct regardless of the choice of \texttt{partition}.
The following theorem shows that under mild assumptions on \texttt{partition}, the expected codelength of A* coding is $\Oh(\KLD{Q}{P})$.
\begin{theorem}[Expected codelength of A* coding]
\label{thm:expected_codelength_of_a_star}
Let $Q$ and $P$ be the target and proposal distributions passed to A* coding (\Cref{alg:a_star_coding}), respectively. Assume that \texttt{partition} satisfies the following property: there exists $\epsilon \in [\nicefrac{1}{2}, 1)$ such that for any node $n$ we have
\vspace{-0.1cm}
\begin{equation}
\label{eq:a_star_bound_condition}
     \Exp[P(B_n)] \leq \epsilon^{D_n}, \vspace{-0.1cm}
\end{equation}
where the expectation is taken over the joint distribution of the samples associated with the ancestor nodes of $n$.
Let $k$ be the node returned by A* coding. Then, we have
\begin{equation}
    \Exp[D_k] \leq -\frac{1}{\log\epsilon}\left[\KLD{Q}{P} + e^{-1} + \log 2 \right].
\end{equation}
In particular, when $\epsilon = 1/2$,
\begin{equation} \label{eq:th1}
    \Exp[D_k] \leq \KLD{Q}{P}_2 + e^{-1}\log_2 e + 1.
\end{equation}
\end{theorem}
\vspace{-0.1cm}
\begin{proof}
See \Cref{section:proof_of_expected_codelength} for proof.
\end{proof}
\vspace{-0.2cm}
Motivated by the results of \Cref{thm:expected_codelength_of_a_star}, we examine two variants of A* coding based on particular choices for \texttt{partition}, which yield AS* and AD* coding.

\textbf{AS* coding:} For node $n$ with triplet $(B_n, X_n, G_n)$, where $B_n = (\alpha, \beta)$, we define \texttt{partition} as
\begin{equation}
\label{eq:as_star_partition_def}
\begin{aligned}
    \partition_{AS^*}(B_n, X_n) &\defeq (\alpha, X_n), (X_n, \beta).
\end{aligned}
\end{equation}
In this case, the following result holds.
\begin{lemma}
\label{lemma:as_star_expected_bound_size}
Let $P$ be the proposal distribution passed to AS* coding and let \texttt{partition} be as defined in \cref{eq:as_star_partition_def}.
Then the condition in \cref{eq:a_star_bound_condition} is satisfied with $\epsilon = 3/4$, that is\
\begin{equation}
    \Exp[P(B_n)] \leq (3/4)^{D_n}.
\end{equation}
\end{lemma}
\begin{proof}
See \Cref{section:proof_of_as_star_bound_size} for proof.
\end{proof}

Hence, the codelength of AS* sampling is bounded by
\begin{equation}
\label{eq:as_star_codelength_bound}
    \lambda~\KLD{Q}{P}_2 + \lambda~\log (\KLD{Q}{P} + 1) + \Oh(1), 
\end{equation}
where $\lambda = -\log2 / \log(3/4) \approx 2.41$.
Further, we have the following result.

\begin{theorem}[Expected runtime of AS* coding]
\label{thm:expected_runtime_of_as_star}
Let $Q$ and $P$ be the target and proposal distributions passed to AS* coding (\Cref{alg:a_star_coding}),
and assume $r = dQ/dP$ is quasiconcave.
Let $T$ be the number of steps AS* takes before it terminates.
Then, we have 
\begin{equation}
    \Exp[T] = \Oh(\infD{Q}{P}) + \Oh(1).
\end{equation}
\end{theorem}
\begin{proof}
See \Cref{section:expected_runtime_of_as_star} for proof.
\end{proof}

\par
\Cref{thm:expected_runtime_of_as_star} identifies a general class of target-proposal pairs where REC and A* sampling can be performed much faster than the $\Omega(\exp(\KLD{Q}{P}))$ and $\Oh(\exp(\infD{Q}{P}))$ bounds shown by \citet{agustsson2020universally} and \citet{maddison2014sampling, maddison2016poisson}, respectively.

\par
\textbf{AD* coding:} Unfortunately, the upper bound on the codelength of AS* coding in \cref{eq:as_star_codelength_bound} is not tight enough to be optimal.
However, \Cref{thm:expected_codelength_of_a_star} suggests this may be addressed by choosing a \partition function for which $\epsilon = \nicefrac{1}{2}$.
One such partition function is dyadic partitioning, which splits $B_n = (\alpha, \beta)$ as
\vspace{-0.1cm}
\begin{equation}
\label{eq:ad_star_partition_def}
\begin{aligned}
    \partition_{AD^*}(B_n, X_n) &\defeq (\alpha, \gamma), (\gamma, \beta), \vspace{-0.4cm}
\end{aligned}
\end{equation}
where $\gamma$ is chosen such that $P((\alpha, \gamma)) = P((\gamma, \beta))$, which is always possible for continuous $P$.
We refer to this partitioning as dyadic, and the corresponding algorithm AD* coding, because for every node $n$ in the search tree with $B_n = (\alpha, \beta)$, we have that $(F_P(\alpha), F_P(\beta))$ forms a dyadic subinterval of $(0, 1)$, where $F_P$ is the CDF of $P$.

\vspace{-0.2cm}
\begin{algorithm}
\SetAlgoLined
\DontPrintSemicolon
\SetKwInOut{Input}{Input}\SetKwInOut{Output}{Output}
\SetKwFunction{pushWithPriority}{push}\SetKwFunction{topPriority}{topPriority}\SetKwFunction{popHighest}{popHighest}\SetKwFunction{isEmpty}{empty}
\Input{Target $Q$, proposal $P$, bounding function $M$, {\color{blue} maximum search depth $D_{max}$}.}
$(LB, X^*, k) \gets (-\infty, \mathrm{null}, 1)$\;
$\Pi \gets \mathrm{PriorityQueue}$\;
${\color{blue}(D_1, H_1) \gets (1, 1)}$\;
$G_1 \sim \TruncGumbel{0}{\infty}$\;
$X_1 \sim P$\;
$M_1 \gets M(\Reals)$\;
$\Pi.\pushWithPriority(1, G_1 + M_1)$\;
\BlankLine
\While{$!~\Pi.\isEmpty()$ and $LB < \Pi.\topPriority()$}{
    $n \gets \Pi.\popHighest()$\;
    $LB_n \gets G_n + (dQ/dP)(X_n)$\;
    \If{$LB < LB_n$}
    {
        $(LB, X^*) \gets (LB_n, X_n)$\;
        {\color{blue}$H^* \gets H_n$}\;
    }
    \BlankLine
    \If{{\color{blue}$D_n \leq D_{max}$}}{
        \BlankLine
        $L, R \gets \partition(B_n, X_n)$
        \BlankLine
        \For{$B \in \{L, R\}$}{
            $(k, B_k) \gets (k + 1, B)$\;
            ${\color{blue}D_k \gets D_n + 1}$\;
            ${\color{blue}H_k \gets \begin{cases}
            2H_n & \text{if } B = L \\
            2H_n + 1 & \text{if } B = R \\
            \end{cases}}$\;
            $G_k \sim \TruncGumbel{\log P(B_k)}{G_n}$\;
            $X_k \sim P\lvert_{B_k}$\;
            \If{$LB < G_k + M_n$}{
                $M_k \gets M(B_k)$\;
                \If{$LB < G_k + M_k$}{
                    $\Pi.\pushWithPriority(k, G_k + M_k)$\;
                }
            }
        }
    }
}
\KwRet{$(X^*, H^*)$}
\caption{
A* coding. {\color{blue} Blue parts} show modifications of A* sampling \citep{maddison2014sampling}.}
\label{alg:a_star_coding}
\end{algorithm}
\vspace{-0.3cm}
We conjecture, that the expected runtime of AD* is $\Oh(\infD{Q}{P})$ and in \Cref{sec:experiments} we provide thorough experimental evidence for this.

\par
\textbf{Depth-limited A* coding:} There is a natural way to set up an approximate REC algorithm based on A*, which takes $\abs{C}$ as an input parameter.
Specifically, we can limit the maximal depth $D_{max}$ to which \cref{alg:a_star_coding} is allowed to search. 
The number of nodes in a complete binary tree of depth $D$ is $2^D - 1$, so setting $D_{max} = \abs{C}$ ensures that each node can be encoded using a heap index with $\abs{C}$ bits.
By limiting the search depth, A* coding returns the optimal node up to depth $D_{max}$ instead of the global optimum. 
However, if we set $D_{max}$ large enough, then depth-limited algorithm should also be able to find the global optimum.
This intuition is made precise in the following lemma.
\begin{lemma}
\label{lemma:ad_star_index_match}
Let $Q$ and $P$ be the target and proposal distributions passed to A* coding (\Cref{alg:a_star_coding}). Let $H^*$ be the heap index returned by unrestricted A* coding and $H^*_d$ be the index returned by its depth-limited version with $D_{max} = d$. Then, conditioned on the public random sequence $S$, we have $H^*_d \leq H^*$. Further, there exists $D \in \Nats$ such that for all $d > D$ we have $H^*_d = H^*$.
\end{lemma}
\vspace{-0.3cm}
\begin{proof}
See \Cref{section:proof_of_ad_star_index_match} for proof.
\end{proof}
\vspace{-0.3cm}
If the depth of the global optimum is larger than $D_{max}$, then depth-limited A* coding will not return an exact sample.
As we reduce $D_{max}$, we force the algorithm to return increasingly sub-optimal solutions, which correspond to more biased samples.
It is therefore important to quantify the trade-off between $D_{max}$ and sample bias.
\Cref{thm:biasedness_of_a_star} bounds the sample bias of depth-limited AD* coding, which we refer to as DAD*, as a function of $D_{max}$.
This result is similar to existing bounds for MRC and ORC.

\begin{theorem}[Biasedness of DAD* coding]
\label{thm:biasedness_of_a_star}
Let $Q$ and $P$ be the target and proposal distributions passed to DAD* coding (\Cref{alg:a_star_coding}). 
Let 
\begin{equation}
    K \defeq \lfloor \KLD{Q}{P}_2 \rfloor,~~D \defeq K + t,~~N \defeq 2^D
\end{equation}
where $t$ is a non-negative integer, $r = dQ/dP$ and $Y \sim Q$.
Let $f$ be a measurable function and define
\vspace{-0.2cm}
\begin{equation}
\norm{f}_Q \defeq \sqrt{\Exp_{Q}[f^2]}.
\vspace{-0.25cm}
\end{equation}
Let $\widetilde{Q}_D$ the distribution of the approximate sample returned by DAD* coding with depth-limit $D$.
Define
\begin{equation*}
    \delta \defeq \left( 2^{-t/4}\sqrt{1 + \frac{1}{N}} + 2\sqrt{\Prob\left[ \log_2 r(Y) > K + t/2 \right]} \right)^{1/2}
\end{equation*}
Then,
\begin{equation}
    \Prob\left[ \abs*{\Exp_{\widetilde{Q}_D}[f] - \Exp_Q[f]} \geq \frac{2\norm{f}_Q \delta}{1 - \delta} \right] \leq 2\delta.
\end{equation}
\end{theorem}
\begin{proof}
See \Cref{section:proof_of_ad_star_bias} for proof.
\end{proof}
\vspace{-0.2cm}
\Cref{thm:biasedness_of_a_star} says that the bit budget for DAD* should be approximately $\KLD{Q}{P}_2$ in order to obtain reasonably low-bias samples. 
As we increase the budget beyond this point, we observe that in practice $\delta$, and by extension the bias, decay quickly. 
In particular, as $t \to \infty$, $\delta \to 0$, and we recover exact AD*. 
We note that it is also possible to depth-limit other variants of A* coding, such as AS*.
However, for any fixed $D_{max}$, the sample bias of these variants will be larger than the bias of DAD*.
This is because, as shown in \cref{thm:expected_codelength_of_a_star}, if we use a different \texttt{partition} with $\epsilon > 1/2$, A* coding will need to search deeper down the tree to find the global optimum.
AD* achieves the lowest possible average depth out of all variants of A* coding, because it has $\epsilon = 1/2$.
Equivalently, AD* can achieve the same sample quality with a lower $D_{max}$ than any other variant of A* coding.
In practice, we observed that depth limited AS* gives significantly more biased samples than AD*, in line with the above reasoning, so we did not further pursue any depth-limited variants other than DAD*.


\par
\textbf{Runtime of DAD* coding:} Based on our conjecture for exact AD* and the result of \Cref{lemma:ad_star_index_match}, we conjecture that DAD* runs in $\Oh(\infD{Q}{P})$ time.
We provide experimental evidence for this in \Cref{sec:experiments}.

\par
\textbf{Using all available codewords:}
When $D_{max} = D$, the number of nodes in the binary search tree of DAD* is $2^D - 1$, which is one fewer than the number of items we can encode in $D$ bits.
In particular, the codeword 0 is never used, since heap indexing starts at 1.
We can make AD* slightly more efficient by drawing 2 samples and arrival times at the root node instead of a single one.
We found that this has a significant effect on the bias for small $D$, and becomes negligible for large $D$.
In \Cref{sec:experiments}, we perform our experiments with this modified version of DAD*.

\par
\textbf{Tying codelengths:} Since the codelength is a parameter of DAD*, we can code samples from different variables using the same codelength, grouping them in a block and passing $D_{max} = |C|$ to \Cref{alg:a_star_coding} for each variable in the block.
Since the variables have the same codelength, we only need to communicate this codelength once per block.
\vspace{-0.3cm}
\section{IsoKL layers and VAEs}
\vspace{-0.1cm}
\par
\textbf{Tying KL divergences:} Although DAD* can be used to tie together the codelengths of different samples, \Cref{thm:biasedness_of_a_star} suggests that the search depth $|C|$ used in DAD* affects the sample bias.
In order to obtain low-bias samples, we must set $|C| \geq \KLD{Q}{P}$.
If we group variables with different KL divergences and code them using the same $|C|$, one of the two following unwanted effects might occur: (1) if $\KLD{Q}{P} \gg |C|$ for a variable, then the corresponding sample will be highly biased; (2) if $\KLD{Q}{P} \ll |C|$ for a variable, then an excessive codelength is being used to code its sample, which is inefficient.
\begin{figure}[h!]
\centering
  \includegraphics[width=0.35\textwidth]{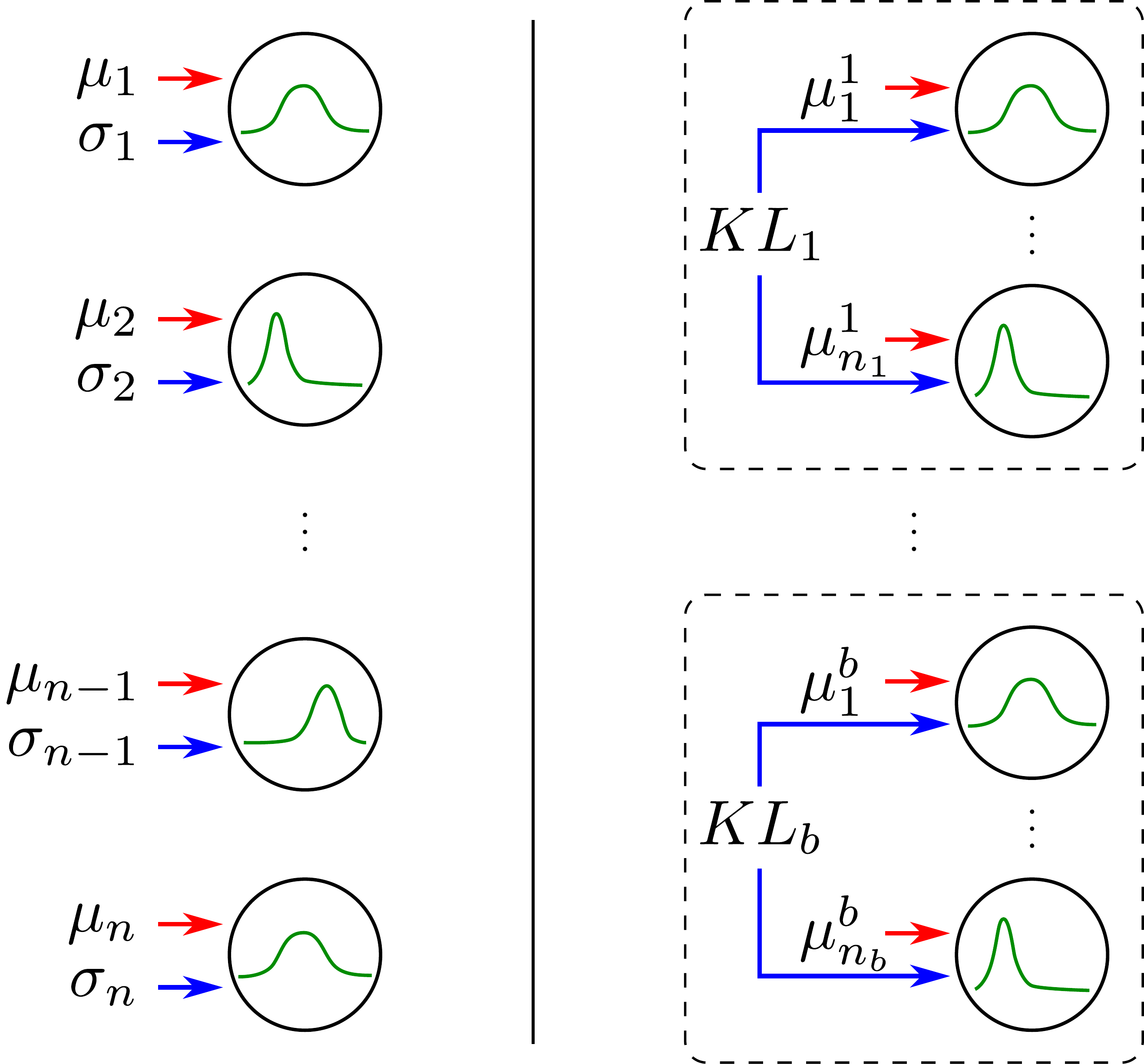}
  \vspace{-0.2cm}
  \caption{Two different ways of parameterizing a Gaussian variational posterior $Q$. \textbf{Left:} The usual mean-variance parameterization. \textbf{Right:} IsoKL layer using our proposed mean-KL parameterization where the latent dimensions are divided into $b$ blocks, and the KLs are shared within each block.}
  \vspace{-0.3cm}
  \label{fig:usual_vs_isokl_layer}
\end{figure}
To avoid these cases, it would be useful if $|C| \approx \KLD{Q}{P}$ for all variables in a block.
We can achieve this by constraining the KL divergences of all variables in a block to be equal.

\par
Suppose $Q$ and $P$ are diagonal Gaussians over $(X_1, \hdots X_N)$ with means $(\mu_1, \hdots, \mu_N)$ and $(\nu_1, \hdots, \nu_N)$, and variances $(\sigma^2_1, \hdots, \sigma^2_N)$ and $(\rho^2_1, \hdots, \rho^2_N)$, respectively.
We can parameterize $Q$ such that \vspace{-0.05cm}
\begin{equation} \label{eq:isokl_param}
    \KLD{\mathcal{N}(\mu_n, \sigma_n^2)}{\mathcal{N}(\nu_n, \rho_n^2)} = \kappa ~ \text{ and } ~ \sigma_n < \rho_n. \vspace{-0.1cm}
\end{equation}
for each $n = 1, \dots, N$, by setting \vspace{-0.05cm}
\begin{align}
    |\mu_n - \nu_n| &< \rho_n \sqrt{2\kappa} \\ \sigma^2_n &= -\rho^2_n W\left( -\exp\left( \Delta_n^2 - 2\kappa - 1 \right) \right), \vspace{-0.05cm}
\end{align}
where $\Delta_n = (\mu_n - \nu_n) / \rho_n$ and $W$ is the principal branch of the Lambert $W$ function \citep{lambert1758observationes}.
Note that the condition that $\sigma_n < \rho_n$ will ensure that $\infD{Q}{P} < \infty$.
While the $W$ is not elementary, it can be computed numerically in a fast and efficient way.
We refer to this as an IsoKL Gaussian layer (see \Cref{fig:usual_vs_isokl_layer}), and call a VAE model using such a $Q$ an IsoKL VAE (IKVAE).
Although we focus on Gaussians,
this approach can be extended to other distributions.
See \Cref{section:iso_kl_layer} for implementation details and mathematical details on deriving the necessary quantities for IsoKL layers.
\vspace{-0.3cm}
\section{Experiments}
\vspace{-0.1cm}
\label{sec:experiments}
\begin{figure*}[h!]
  \vspace{-0.1cm}
  \includegraphics[width=\textwidth]{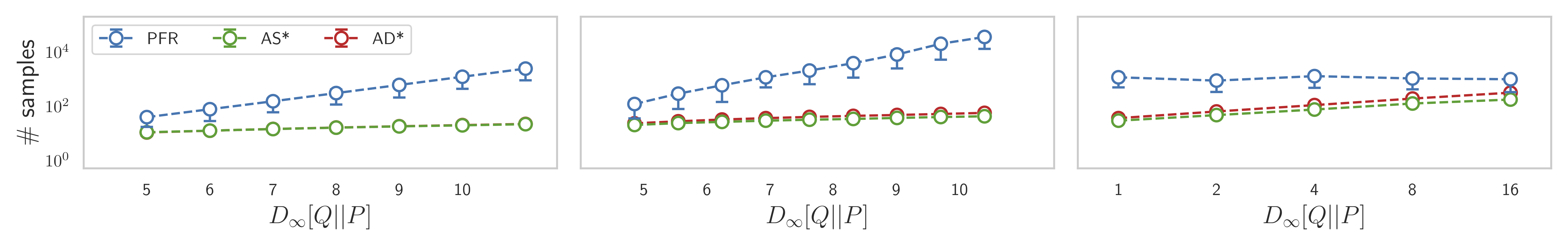} \vspace{-0.3cm}
  \includegraphics[width=\textwidth]{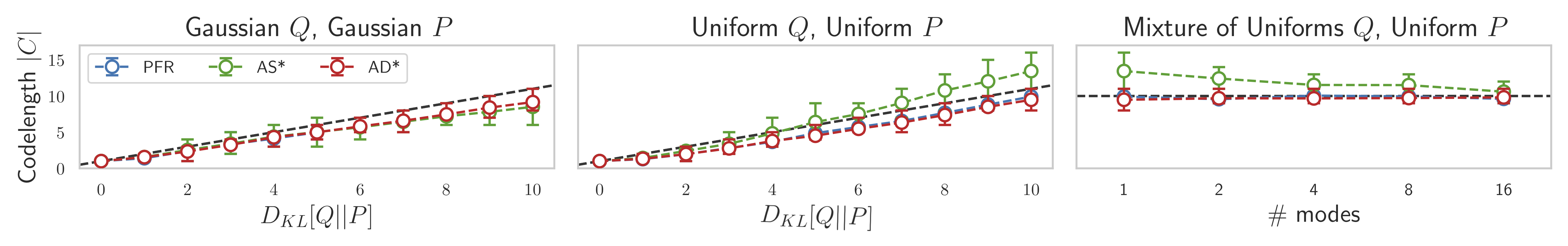}
  \vspace{-0.6cm}
  \caption{Number of steps (top) and codelength (bottom) for PFR, AS* and AD* coding.
  Reported codelengths do not include the overhead terms.
  Circles show mean values, and error bars show first and third quantiles.
  AS* and AD* are significantly quicker than PFR.}
  \label{fig:exact_rec}
  \vspace{-0.4cm}
\end{figure*}

\begin{figure*}[h!]
  \includegraphics[width=\textwidth]{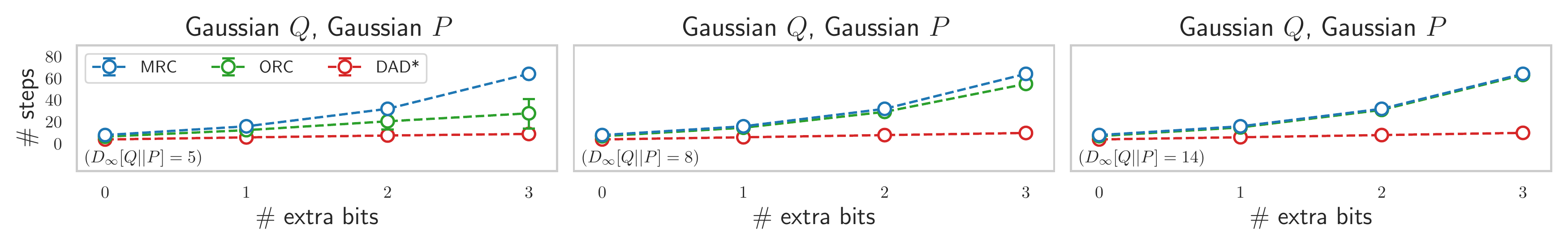} \vspace{-0.3cm}
  \includegraphics[width=\textwidth]{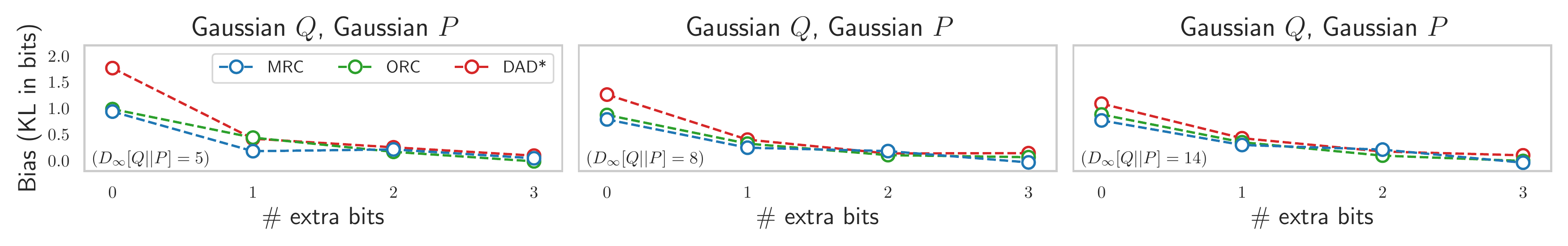}
  \vspace{-0.6cm}
  \caption{Number of steps (top) and bias (bottom) for MCR, OCR and DAD* coding.
  The bias is measured as the KL divergence from the empirical distribution of samples, to the true target (discussion in text). 
  DAD* has similar bias to MRC and ORC, but is much faster.}
  \label{fig:approx_rec}
  \vspace{-0.5cm}
\end{figure*}
\vspace{-0.05cm}
\textbf{Experiments for exact REC:} We conducted experiments using PFR, AS* and AD* coding to perform REC on Gaussian, uniform and disjoint mixture of uniform distributions, where we systematically vary the problem parameters.
\vspace{-0.05cm}
\par
\Cref{fig:exact_rec} shows measurements for PFR, AS* and AD* coding.
We report the number of steps executed by each algorithm rather than wall-clock time, as the former is proportional to the runtime and is unaffected by specific differences in implementation.
We also report the codelength, in bits, excluding the additional logarithmic and constant overhead.
First, we observe that the number of steps taken by PFR scales exponentially with $\infD{Q}{P}$, as expected.
This scaling renders PFR computationally intractable for practical problems.
By contrast, the number of steps taken by AS* and AD* coding increases linearly with $\infD{Q}{P}$.
\vspace{-0.05cm}
\par
Second, we observe that the number of modes has an effect on the runtime of both AS* and AD* coding.
In the top right of \Cref{fig:exact_rec} we have set $Q$ to a mixture of uniforms with disjoint support, $P$ to a uniform, and fixed $\infD{Q}{P}$ while varying the number of modes in $Q$.
For a small number of modes, AS* and AD* are both significantly faster than PFR.
As the number of modes increases, so does the number of steps executed by AS* and AD*.
This trend is expected, because for unimodal $Q$, A* coding can quickly bound the objective for large regions of the input space which have low $dQ/dP$ throughout.
As $dQ/dP$ becomes increasingly multimodal, the bounding function $M$ becomes large in increasingly many disconnected regions of the search space.
Thus both AS* and AD* must drill down to and search increasingly many of these regions, before producing a sample, which requires a larger number of steps.
By contrast, PFR is unaffected by the number of modes in $dQ/dP$, since it is equivalent to Global Bound A*, which retains the entire input space in a single active search branch.
These results suggest that the unimodality of $dQ/dP$ is a key attribute for enabling fast coding with AS* and AD*.
\vspace{-0.05cm}
\par
Third, we observe that the codelengths of PFR, AS* and AD* scale linearly with $\KLD{Q}{P}$, as expected.
However, in some cases AS* produces a larger mean codelength than PFR and AD*.
This can be explained by the fact that the dominant term in the codelength bounds of the three algorithms is $\lambda~ \KLD{Q}{P}$, where PFR and AD* have $\lambda = 1$, while AS* has $\lambda \approx 2.41$.
AS* can produce larger $|C|$ than AD* because, as \Cref{lemma:as_star_expected_bound_size} states, the expected rate of shrinkage of its search region is slower than AD*.
Therefore, in order to refine its search by the same amount, AS* needs a greater number of steps, leading to a larger expected $|C|$.

\par
\textbf{Experiments for approximate REC:} We conducted experiments using MRC, ORC and DAD* coding to perform approximate REC with Gaussian $Q$ and $P$, varying the bit budget that is allowed to the coders, in addition to the baseline budget of $\KLD{Q}{P}_2$ bits.
\Cref{fig:approx_rec} shows the effect of the additional bit budget on the number of steps and the sample bias for three problems with different $\infD{Q}{P}$.
We quantify the sample bias as the KL divergence, $\KLD{\hat{Q}}{Q}_2$, from the empirical distribution $\hat{Q}$ of the approximate samples, to the true $Q$.
In each case, we draw $100$ samples and follow the method of \citet{perez2008kullback} to estimate $\KLD{\hat{Q}}{Q}_2$.
We observe that in all three cases (\cref{fig:exact_rec} bottom), while there is a slight difference in the level of bias when no extra bits are allowed, the difference in bias becomes negligible when one or more extra bits are added.
However, DAD* achieves this bias far faster than MRC and ORC (\cref{fig:exact_rec} top), making it a far more tractable method.
\setlength\tabcolsep{4.0pt}
\vspace{-0.2cm}
\begin{table}[h!] \small
\begin{center}
\begin{tabular}{ l c c c c } 
 \toprule
 & {\scshape \# latent} & {\scshape NEG. ELBO} & {\scshape AD*} & {\scshape DAD*} \\ \midrule
 \multirow{ 2}{*}{{\scshape VAE}} & 20 & {\scriptsize $1.43 \pm 0.01$} & {\scriptsize $1.53 \pm 0.01$} & $-$  \\
  & 50 & {\scriptsize $1.40 \pm 0.01$} & {\scriptsize $1.66 \pm 0.01$} & $-$ \\ \midrule
 \multirow{ 2}{*}{{\scshape IKVAE}} & 20 & {\scriptsize $1.44 \pm 0.01$} & {\scriptsize $1.55 \pm 0.01$} & {\scriptsize $1.51 \pm 0.01$} \\
  & 50 & {\scriptsize $1.44 \pm 0.01$} & {\scriptsize $1.69 \pm 0.01$} & {\scriptsize $1.60 \pm 0.01$} \\
 \bottomrule
\end{tabular}
\end{center}
\vspace{-0.2cm}
\caption{ELBO and lossless compression rates (bpp) on MNIST.}
\label{tab:mnist1}
\end{table}
\setlength\tabcolsep{5.0pt}
\begin{table}[h!] \vspace{-0.4cm} \small
\begin{center}
\begin{tabular}{ l c c c } 
 \toprule
 & {\scshape \# latent} &  {\scshape AD*} & {\scshape DAD*} \\ \midrule
 \multirow{ 2}{*}{{\scshape IKVAE}} & 20 & {\scriptsize $86.31 \pm 0.01$} & {\scriptsize $5.00 \pm 0.01$}  \\
  & 50 & {\scriptsize $195.20 \pm 0.01$} & {\scriptsize $5.00 \pm 0.01$} \\
 \bottomrule
\end{tabular}
\end{center}
\vspace{-1em}
\caption{Overhead codelength (bits) on MNIST.}
\label{tab:mnist2}
\end{table}
\vspace{-0.4cm}
\par
\textbf{Image compression MNIST:} We compared the performance of AD* and DAD* on image compression experiments on MNIST, using the feedforward VAE architecture of \citet{townsend2018practical}, with Gaussian $Q$ and $P$.
We also trained IKVAEs with the same architecture, using an IsoKL Gaussian $Q$.
As \Cref{thm:expected_codelength_of_a_star} shows and as we discuss above, AS* has strictly worse expected codelength than AD* and hence we did not include it in our experiments.
For DAD*, we set $\kappa = 2$ based on preliminary experiments.
\Cref{tab:mnist1} shows the lossless compression rates of different model and coder combinations, from which we observe the following trends.
First, the IKVAEs achieve similar ELBOs to the VAEs, suggesting that tying the KLs does not degrade model performance.
Further, we observe that for the IKVAE architectures, using DAD* improves the compression rate over AD*.
We do not provide results for DAD* applied to a standard VAE posterior,
as it would yield a strictly worse performance than AD*.
This is because the KLs in each latent dimension are different, and hence not only do we need to communicate the codelength for each dimension, but DAD* returns approximate samples from the target distribution, as opposed to AD*.
This is because, as corroborated by \cref{tab:mnist2}, DAD* significantly reduces codelength overheads, improving performance over AD*.
We also note that in this task, the number of latent dimensions is relatively small (20 or 50) compared to the number of image pixels (784).
Since the overhead costs increase with the number of latent variables, we expect the savings of DAD* to be more pronounced for larger IKVAEs.
Overall, this experiment demonstrates that AD* can be effective for VAEs, while DAD* and IKVAEs further reduce coding overheads.
\section{Related Work}
\vspace{-0.1cm}
\par
\textbf{Quantization-based approaches:} Most state-of-the-art methods in both lossy and lossless compression are based on including quantization in the compression pipeline, and somehow circumventing its non-differentiability during training.
Current widespread approaches in lossy compression use VAEs with a particular choice of latent distributions \citep{balle2017end, balle2018variational}.
Instead of quantizing latent representations during training, these methods perturb the latents with uniform noise, a technique known as dithering.
To perform compression, the methods switch back to hard quantization.
Dithering is equivalent to using a uniform variational posterior distribution, and has been demonstrated to work well in practice \citep{balle2020nonlinear}. 
However, this introduces a mismatch between the training and compression phases and it also constrains the design choices for new methods to use uniform distributions.
A related variant is the work of \citet{agustsson2020universally}, who propose to use universal quantization \citep{ziv1985universal} for compression.
Universal quantisation can be regarded as a REC algorithm.
While this method performs very well in the experiments of \citet{agustsson2020universally}, it is limited to particular choice of distribution.
Our work can be regarded as a step towards lifting the restrictions imposed by quantization based methods.

\par
\textbf{Bits-back coding:}
\citet{townsend2018practical} introduced a practical way to combine bits-back coding \citep{hinton1993keeping} and VAEs to perform lossless data compression. 
The method of \citet{townsend2018practical} has later been applied to normalizing flows \citep{ho2019compression}, which, together with discrete \citep{van2020idf++} and quantization based approaches \citep{zhang2021iflow} represent the state-of-the-art in lossless image compression. 
However, bits-back coding is only applicable to perform lossless compression.
Furthermore, it is only asymptotically efficient, meaning that it has a large constant codelength overhead that becomes insignificant as we compress larger batches of data.
In contrast, our method is applicable to both lossy and lossless compression and can be used to perform one-shot compression.

\par
\textbf{REC and reverse channel coding:} REC was first proposed by \citet{havasi2018minimal}, who developed MRC for compressing Bayesian neural networks, and was later extended to data compression using VAEs by \citet{flamich2020compressing}.
Both of these works were inspired by reverse channel coding \citep{bennett2002entanglement}.
REC and reverse channel coding can be viewed as the worst-case and average-case approaches to the same problem.
Concretely, REC requires that for fixed proposal $P$ and public randomness $S$, the codelength bound in \cref{eq:rec_codelength_constraint} holds for any target $Q$.
In contrast, reverse channel coding assumes a distribution over a family of possible targets and requires that \cref{eq:rec_codelength_constraint} holds in expectation.

The first general REC algorithm for discrete distributions was proposed by \citet{harsha2007communication}, however this method has an impractically long runtime.
\citet{li2018strong} developed PFR coding, an alternative based on Poisson processes, however this also is computationally impractical.
Recently, \citet{theis2021algorithms} proposed ORC, which combines some the benefits of MRC and PFR, but remains computationally impractical.
These algorithms all share the limitation that their expected runtime is $\Omega(\exp(\KLD{Q}{P}))$.
\par
\textbf{Hybrid Coding:}
To improve the runtime of REC/RCC algorithms, \citet{theis2021algorithms} proposed hybrid coding (HC).
HC is applicable whenever the support of the target is compact, and can be combined with any of the existing REC/RCC algorithms to improve their runtime by a multiplicative factor.
However, HC does not change the asymptotic complexity of the REC/RCC algorithm that it is combined with.
We note that while \citet{theis2021algorithms} used HC in conjunction with ORC, HC can equally well be combined with A* coding.
Thus the speedup that HC provides is complementary to that of A* coding.

\par
\textbf{Efficient rejection sampling:} 
While rejection sampling is a applicable to any target $Q$ and proposal $P$ where $\infD{Q}{P} < \infty$ and $dQ/dP$ can be evaluated, the basic algorithm has $\Oh(\exp(\infD{Q}{P}))$ expected runtime complexity \cite{maddison2016poisson}. 
Thus, a natural question is to ask under what assumptions it is possible to perform rejection sampling from a target $Q$ using a proposal $P$ efficiently.
\par
Recently, \citet{chewi2021rejection} consider the problem of constructing an upper envelope for rejection sampling from discrete probability distributions.
In particular, they study the time complexity of constructing the envelope as a function of the alphabet size.
They show that shape constraints on the target distribution, such as monotonicity and log-concavity can be utilized to design algorithms whose runtime scales logarithmically in the alphabet size.
Our work is complementary to theirs, as \Cref{thm:expected_runtime_of_as_star} provides an initial result showing that shape constraints can be leveraged to design more efficient sampling algorithms for continuous distributions as well.
\vspace{-0.3cm}
\section{Conclusion}
\vspace{-0.1cm}
\textbf{Summary:} In this work we proposed AS* coding and AD* coding, two algorithms based on A* sampling for performing REC with one-dimensional target and proposal distributions $Q$ and $P$.
We proved that the expected codelengths of AS* and AD* are $\Oh(\KLD{Q}{P})$ and that, whenever $dQ/dP$ is unimodal, the expected runtime of AS* is $\Oh(\infD{Q}{P})$.
Experimental evidence suggests that the runtime of AD* is also $\Oh(\infD{Q}{P})$.
This runtime significantly improves upon the existing runtimes of existing REC algorithms, without placing severe conditions on $Q$ and $P$.

In addition, we proposed two methods to eliminate overhead codelength when encoding multiple samples.
First, we introduced an approximate depth-limited variant of AD* coding, DAD* coding, in which the codelength of the encoder is a tunable variable, and proved an upper bound for its bias.
Second, we introduced the IsoKL parameterization, in which latent dimensions are grouped into blocks, and the KL divergences of all latent dimensions in a block are constrained to be equal.
DAD* together with the IsoKL parameterization allow us to encode multiple samples with the same codelength, thereby amortising coding overheads, while maintaining low sample bias.
Experimentally, we demonstrated the favourable runtimes of AS* and AD* coding on extensive toy experiments.
We also showed that DAD* coding achieves levels of bias comparable to existing approximate REC algorithms, while maintaining a significantly faster runtime.
On lossless image compression experiments on MNIST, DAD* together with an IsoKL VAE (IKVAE) parameterization achieved a compression rate close to the theoretically optimal ELBO.

\textbf{Further work:} One of the central remaining questions of this work is the runtime of AD* coding.
Based on our experiments, we conjecture that the runtime of AD* coding is also $\Oh(\infD{Q}{P})$ whenever $dQ/dP$ is unimodal, however this remains to be shown.
\par
In general, for fixed $\KLD{Q}{P}$ we can have arbitrarily high $\infD{Q}{P}$, hence a second, more general question is if there exists a REC algorithm with $\Oh(\KLD{Q}{P})$ expected runtime.
Adaptive rejection sampling \citep{gilks1992adaptive}, OS* sampling \citep{dymetman2012algorithm} and ideas from \citep{theis2021algorithms} could be good starting points for developing such an algorithm.
\par
Another promising direction is to apply the IsoKL parameterization to larger VAEs, such as those used by \citet{townsend2019hilloc}, and scale our approach up to real-world compression tasks.
Lastly, our methods can also be readily applied to lossy compression.

\section{Author Contributions}

GF discovered that A* sampling can be modified to obtain A* coding (\cref{alg:a_star_coding}), which can be used to perform relative entropy coding, and provided proofs for \cref{thm:expected_codelength_of_a_star,thm:biasedness_of_a_star}.
SM provided a proof for \cref{thm:expected_runtime_of_as_star}.
GF and SM contributed equally to the experiments and the writing of this paper.
JMH supervised and steered the project.

\section{Acknowledgements}

We would like to thank Rich Turner and Lennie Wells for useful feedback on an early manuscript of this paper.
GF acknowledges funding from DeepMind.
SM acknowledges funding from the Vice Chancellor's \& George and Marie Vergottis scholarship of the Cambridge Trust.



\bibliography{main}
\bibliographystyle{icml2021}


\newpage
~\
\newpage

\appendix
\onecolumn
\section{Proof of \Cref{thm:expected_codelength_of_a_star}}
\label{section:proof_of_expected_codelength}
\paragraph{Notation: }
Throughout the appendix, we will write $[a:b] \defeq [a, b] \cap \Nats$ for $a, b \in \Nats$ and $[a] \defeq [1:a]$. Furthermore, for a vector we will write $x_{1:n} \defeq (x_1, \hdots x_n)$ for $n \in \Nats$.

\par
We make use of the top-down construction of Gumbel processes \citep[\Cref{alg:priority_queue_construction};][]{maddison2014sampling}. 
The top-down construction realizes samples from a base distribution $P$ along with their associated Gumbel values, which together form a Gumbel process. 
A Gumbel process can be thought of as a generalization of the Gumbel-Max trick \citep{papandreou2011perturb}, where the log-probability of each member of a sample space is perturbed with i.i.d.\ Gumbel noise.
Then, it can be shown that the maximum of this process is Gumbel distributed, and the argmaximum has law $P$.
Gumbel processes can be shown to be equal in distribution to exponential races, where the time variable is mapped to its negative logarithm \citep{maddison2016poisson}. 
This is important, as it allows us to switch between the Gumbel and Poisson process representations, to leverage existing results in our analysis.
\par
\Cref{alg:priority_queue_construction} realizes its Gumbel process using a space partitioning binary tree construction while also recording the depths and heap indices of nodes. 
It is therefore an extension of Algorithm 1 in \citet{maddison2014sampling}.
\Cref{alg:priority_queue_construction} can be realized using public randomness by anyone with access to the public seed $S$.
\par
A* sampling can be viewed as performing a binary tree search on the Gumbel process with measure $P$, as realized by \Cref{alg:priority_queue_construction}, to search for a sample with distribution $Q$. 
The key observation is that for the search to proceed, the whole realization of the Gumbel process with measure $P$ is not needed, and in fact it can be realized on-the-go.
\par
A* coding (\Cref{alg:a_star_coding}) first runs the regular A* sampling procedure by simulating the Gumbel process with the proposal measure $P$ using the publicly available randomness $S$. 
Then, it encodes the returned sample using the heap index $H_n$ of the node $n$ with which the sample is associated.
Since any node with a given heap index can be simulated without reference to $Q$ using \Cref{alg:priority_queue_construction}, the code returned by A* coding is always uniquely decodable given $S$, and the correctness of A* sampling \citep{maddison2014sampling} will guarantee that the sample A* coding returns has the correct distribution $Q$.
\par
PFR and ORC also operate on Gumbel/Poisson processes, however, they use a different encoding process. 
They encode the \textit{index} $K$ of a sample as opposed to its heap index $H$.
For a sample $x$, $K$ is obtained by sorting the arrival times in the Gumbel/Poisson process with measure $P$, and returning the index associated with $x$ in the sorted list.
PFR and ORC obtain $K$ easily, because they realize the Gumbel/Poisson process in-order (for example, see Algorithm 3 in \citep{maddison2014sampling}). 
Similarly, \Cref{alg:priority_queue_construction} also constructs the process in order, however, it uses the top-down construction.

\par
\Cref{thm:expected_codelength_of_a_star} shows that A* coding is not only correct and uniquely decodable, but its expected codelength is also optimal. 
However, to show this, we first show the following intermediate result, which relates the index $K$ of a sample to its expected depth $D$ in the top-down construction.

\begin{lemma}[Average depth of nodes in a Gumbel process]
\label{lemma:average_depth_of_nodes_in_prior_process}
Let $P$ be a Borel probability measure over some Polish space that is supplied to \Cref{alg:priority_queue_construction}. Assume that  \texttt{partition} satisfies the following property: there exists $\epsilon \in [\nicefrac{1}{2}, 1)$, such that for any node $n$ we have
\begin{equation}
    \Exp[P(B_n)] \leq \epsilon^{D_n},
\end{equation}
where the expectation is taken over the joint distribution of the samples associated with the ancestor nodes of $n$. Let $K_n$ be the index of a node $n$ realized by \Cref{alg:priority_queue_construction} and let $D_n$ be its depth in the tree. Then,
\begin{equation}
    \Exp[D_n \mid K_n] \leq -\log_{\epsilon}K_n.
\end{equation}
\end{lemma}
Before delving into the proof, we clarify the allowed domain of $\epsilon$. First note that $\epsilon$ is solely a property of \texttt{partition}. Note that fixing $\epsilon = 1$ would imply that the bounds do not shrink. The reason why $\epsilon \geq \nicefrac{1}{2}$ is because if \texttt{partition} breaks some region $B$ into $L$ and $R$ with $P(L) = \epsilon_L P(B)$, then necessarily $P(R) = \epsilon_R P(B) = (1 - \epsilon_L) P(B)$, hence we need to take $\epsilon = \max\{\epsilon_L, \epsilon_R\}$. This means that the minimal $\epsilon$ is achieved when $\epsilon_L = \epsilon_R = \nicefrac{1}{2}$. Finally,
since \texttt{partition} might depend on the samples drawn up to reaching the node with bound $B$, we need to take expectation over these.

\begin{proof}
\par
Let $\Frontier_k = (f_1, \hdots f_k)$ be the \textit{frontier} of \cref{alg:priority_queue_construction} after $k$ steps. We shall say that \cref{alg:priority_queue_construction} or \cref{alg:a_star_coding} \textit{expand} a node, which means that they pop off the highest priority node from their priority queue.
$\Frontier_k$ consists of all the nodes that could be expanded, starting with $\Frontier_1$ only containing the root node.
In the search literature $\Frontier_k$ is also commonly referred to as the open set of nodes.
A simple inductive argument shows, that for a binary tree on $k$ nodes will always have $k + 1$ nodes on its frontier, i.e.\ $\abs{\Frontier_k} = k$.
\par
Let $\{(X_i, G_i)\}_{i = 1}^{k - 1}$ be the samples and Gumbels realized by \cref{alg:priority_queue_construction}, sorted in descending order by the $G_i$s up to the $k - 1$ largest one. Let $\mu_f = \log P(B_f)$ be the location parameter of the truncated Gumbel variate $G_f$ for a node $f \in \Frontier_k$.
\citet{maddison2014sampling} show (see their Appendix, the section titled ``Equivalence Under \texttt{partition}''), that regardless of the choice of \texttt{partition}, 
\begin{equation}
    \forall f \in \Frontier_k\quad G_f \sim \TruncGumbel{\mu_f}{G_{k - 1}}.
\end{equation}
Let $F_k \in \Frontier_k$ denote the node that is expanded in step $k$, i.e.\
\begin{equation}
    F_k \sim \argmax_{f \in \Frontier_k}\{ G_f \}.
\end{equation}
Then a simple Gumbel-max trick-type argument shows \citep{maddison2014sampling}, that
\begin{equation}
\begin{aligned}
    p(F_k = f \mid G_{1:k - 1}, X_{1:k - 1}) &= \frac{\exp(\mu_f)}{\sum_{\phi \in \Frontier_k}\exp(\mu_\phi)} \\
    &= \frac{\exp(\log P(B_f))}{\sum_{\phi \in \Frontier_k}\exp(\log P(B_\phi))} \\
    &= \frac{P(B_f)}{\sum_{\phi \in \Frontier_k} P(B_\phi)} \\
    &= P(B_f),
\end{aligned}
\end{equation}
where the last equality holds because the bounds associated with the nodes on $\Frontier_k$ form a partition of the whole sample space for any $k$. Then,
\begin{equation}
\begin{aligned}
    \Exp_{p(F_k = f \mid G_{1:k - 1}, X_{1:k - 1})}[-\mu_f] &= -\sum_{f \in \Frontier_k}P(B_f) \mu_f \\
    &= -\sum_{f \in \Frontier_k} P(B_f) \log P(B_f) \\
    &= H[F_k] \\
    &\leq \log k.
\end{aligned}
\end{equation}
The last equality follows, since the maximal Shannon entropy of a distribution over $k$ items is the uniform distribution with entropy $\log k$.

Now, taking expectations over $G_{1:k - 1}, X_{1:k - 1}$, we get
\begin{equation}
    \Exp_{p(F_k = f, G_{1:k - 1}, X_{1:k - 1} \mid k)}[-\mu_f] \leq \log k.
\end{equation}
Finally,
\begin{equation}
\begin{aligned}
    \log k &\geq \Exp_{p(F_k = f, G_{1:k - 1}, X_{1:k - 1} \mid k)}[-\mu_f] \\
    &= \Exp_{p(F_k = f, G_{1:k - 1}, X_{1:k - 1} \mid k)}[-\log P(B_f)] \\
    &\geq -\log \Exp_{p(F_k = f, G_{1:k - 1}, X_{1:k - 1} \mid k)}[P(B_f)] \\
    &\geq -\log \epsilon^{D_k} \\
    &= -D_k \log \epsilon,
\end{aligned}
\end{equation}
where the second inequality holds by Jensen's inequality and the third inequality holds by our assumption on \texttt{partition}. Since $0 < -\log \epsilon$, rearranging the two sides of the inequality gives the desired result.

\end{proof}

\par
With this result in mind, we are now ready to prove \Cref{thm:expected_codelength_of_a_star}, which we state again for completeness:

\begin{theorem}[Expected codelength of A* coding]
Let $Q$ and $P$ be the target and proposal measures passed to A* coding (\Cref{alg:a_star_coding}). Assume that \texttt{partition} satisfies the following property: there exists $\epsilon \in [\nicefrac{1}{2}, 1)$ such that for any node $n$ we have
\begin{equation}
\label{eq:a_star_bound_condition_appendix}
     \Exp[P(B_n)] \leq \epsilon^{D_n},
\end{equation}
where the expectation is taken over the joint distribution of the samples associated with the ancestor nodes of $n$.
Let $k$ be the node returned by A* coding. Then, we have
\begin{equation}
    \Exp[D_k] \leq -\frac{1}{\log\epsilon}\left[\KLD{Q}{P} + e^{-1} + \log 2 \right].
\end{equation}
In particular, when $\epsilon = 1/2$,
\begin{equation}
    \Exp[D_k] \leq \KLD{Q}{P}_2 + e^{-1}\log_2 e + 1.
\end{equation}
\end{theorem}

\begin{proof}
Let $K$ be the random variable that represents the index (not the heap index) of the sample returned by A* coding.
\citet{maddison2016poisson} show that $K$ is equal in distribution to the index returned by running global bound A* sampling, which is equal in distribution to the index returned by PFR coding. 
In \citet{li2018strong}, Appendix, section A it is shown, that
\begin{equation}
    \Exp[\log K] \leq \KLD{Q}{P} + e^{-1} + \log 2.
\end{equation}
Putting this together with \Cref{lemma:average_depth_of_nodes_in_prior_process}, we get the desired result.
\end{proof}

\section{Proof of \Cref{lemma:as_star_expected_bound_size}}
\label{section:proof_of_as_star_bound_size}

In the main text, the result is stated for AS* sampling. However, the result does not depend on the target $Q$, only the realization of the Gumbel process with base measure $P$ in \cref{alg:priority_queue_construction}, hence we restate the lemma here as follows:

\begin{lemma}
\label{lemma:as_star_expected_bound_size_appendix}
Let $P$ be a non-atomic proposal measure over a 1-dimensional sample space, passed to \Cref{alg:priority_queue_construction}, and let \texttt{partition} be as defined in \cref{eq:as_star_partition_def}.
Then the condition in \cref{eq:a_star_bound_condition} is satisfied with $\epsilon = 3/4$, that is\
\begin{equation}
    \Exp[P(B_n)] \leq \left(\frac{3}{4}\right)^{D_n}.
\end{equation}
\end{lemma}
\begin{proof}
Let $F$ denote the CDF of the measure $P$. 
we will prove the claim by induction. For the base case, note, that depth $D = 0$ can be associated with not having drawn any samples yet, i.e. the next node that is expanded by \cref{alg:priority_queue_construction} will be the root node. Since the sample is drawn from the whole space, we will have $P(B_1) = P(\Omega) = 1 = (3/4)^0$. For the hypothesis, assume the claim holds for $D = d$. Let $D = d + 1$. Fix a node $n$ such that $D_n = d + 1$. Let $\Ancestors(n) = (n_1, \hdots, n_{D_n - 1})$ denote the ancestors of $n$, where $n_1$ is the root of the tree and $n_{D_n - 1} = \parent(n)$ is the direct parent node of $n$ in the tree constructed by \cref{alg:priority_queue_construction}. Then, by the law of iterated expectations, we have
\begin{equation}
\label{eq:expected_as_star_mass}
\begin{aligned}
    \Exp_{p(X_{\Ancestors(n)})}[P(B_n)] &= \\ \Exp_{p(X_{\Ancestors(\parent(n))})}
    &[\Exp_{p(X_{\parent(n)} \mid X_{\Ancestors(\parent(n))})}[P(B_n)] ].
\end{aligned}
\end{equation}
Focusing on the inner expectation, let $B_{\parent(n)} = (a, b)$. Then, by the definition of \texttt{partition}, $L = (a, X_n), R = (X_n, b)$. Note, that since $B_n$ is either $L$ or $R$, we have $P(B_n) \leq \max\{P(L), P(R)\}$. Furthermore, since the space is 1 dimensional, we get $P(L) = F(X_n) - F(a)$ and $P(R) = F(b) - F(X_n)$. Let $\Uniform{c}{d}$ denote the uniform density on $(c, d)$. Let $\alpha = F(a), \beta = F(b)$. Then, by the generalized probability integral transform, we find, that $U \defeq F(X_n) \sim \Uniform{\alpha}{\beta}$.
Thus, by the law of the unconscious statistician, the inner expectation of \cref{eq:expected_as_star_mass} can be rewritten as 
\begin{equation}
\begin{aligned}
    \Exp_{p(X_{\parent(n)} \mid X_{\Ancestors(\parent(n)))}}&[P(B_n)] \\
    &\leq\Exp_{p(U)}[\max\{U - \alpha, \beta - U\}] \\
    &=\int_\alpha^\beta \frac{\max\{u - \alpha, \beta - u\}}{\beta - \alpha} du \\
    &=\frac{3}{4}(\beta - \alpha) \\
    &=\frac{3}{4}P(B_{\parent(n)}).
\end{aligned}
\end{equation}
Now, by the induction hypothesis \cref{eq:expected_as_star_mass} becomes
\begin{equation}
    \frac{3}{4}\Exp_{p(X_{\Ancestors(\parent(n))})}[P(B_n)] \leq \left(\frac{3}{4} \right)^{d + 1},
\end{equation}
which concludes the proof.
\end{proof}

\section{Proof of \Cref{thm:expected_runtime_of_as_star}}
\label{section:expected_runtime_of_as_star}

\textbf{Note:}
Our original argument for the linear runtime of A* coding contained an error.
\citet{markou2022notes} provided a proof for \cref{thm:expected_runtime_of_as_star}, which we reproduce here.

\textbf{Overview:}
The proof breaks down the execution of AS* coding into two stages.
For the first stage, we consider how AS* shrinks its search bounds, until it obtains a sufficiently good candidate sample.
Here, a sufficiently good sample is a sample which falls within a predefined super-level set of the density ratio.
\Cref{lemma:n_bound} gives an upper bound on the expected number of steps in this first stage of the algorithm.

For the second stage, we quantify how many additional steps AS* must subsequently make until it terminates, after obtaining a good candidate sample in the first stage.
\Cref{lemma:k_bound} gives an upper bound on the expected number of steps in this second stage of the algorithm.
Putting \cref{lemma:n_bound,lemma:k_bound} together, we obtain an upper bound on the runtime of AS*, stated in \cref{corollary:t_bound}.
This bound depends on, and holds for any, super-level set.
Therefore, we can minimise this bound over all super-level sets of the density ratio.
Lastly, we show that even for the worst case density ratios of this bound, this minimum results in a runtime is linear in the $\infty$-divergence, resulting in \cref{thm:expected_runtime_of_as_star}.

\textbf{Notation:}
In this section, all indices to random variables are integers corresponding to the depth of the variable within the binary tree being searched.
This is in contrast to other sections where the random variables are indexed by the node of the binary tree to which they belong.
We found this notational overloading makes the exposition clearer, and the meaning of the indexing should be clear from the context. \\

\begin{assumption}[Continuous distributions, finite $D_{\infty}$] \label{assumption:runtime_proof}
We assume that measures $Q$ and $P$ describe continuous random variables, so their densities $q$ and $p$ exist.
Since $P \gg Q$, the Radon-Nikodym derivative $r(x) = (dQ/dP)(x)$ also exists.
We also assume $r(x)$ is unimodal and satisfies
\begin{equation}
    \Dinf{Q}{P} = \log \sup_{x\in \mathbb{R}} \frac{dQ}{dP}(x) = \log r_{max} < \infty.
\end{equation}
\end{assumption}

Without loss of generality, we can also assume $P$ to be the uniform measure on $[0, 1]$, as shown by the next lemma.
This is because we can push $P$ and $Q$ through the CDF of $P$ to ensure $P$ is uniform, while leaving the Radon-Nikodym derivative unimodal and the $\infty$-divergence unchanged. \\

\begin{lemma}[Without loss of generality, $P$ is uniform]
Suppose $Q$ is a target measure and $P$ a proposal measure as specified in Assumption \ref{assumption:runtime_proof}.
Let $\Phi$ be the CDF associated with $P$ and consider the measures $P', Q' : [0, 1] \to [0, \infty)$ defined as
\begin{equation}
    P' = P \circ \Phi^{-1} ~\text{ and }~ Q' = Q \circ \Phi^{-1}.
\end{equation}
Then, $P'$ is the uniform measure on $[0, 1]$.
Further, the Radon-Nikodym derivative $dQ'/dP'(x)$ is unimodal, and
\begin{equation}
    \log \sup_{z \in [0, 1]} \frac{dQ'}{dP'}(z) = \log \sup_{x \in \mathbb{R}} \frac{dQ}{dP}(x).
\end{equation}
\end{lemma}
\begin{proof}
    First, $P'$ is the uniform measure on $[0, 1]$ since for any $z \in [0, 1]$
    \begin{equation}
        P'([0, z]) = P \circ \Phi^{-1}([0, z]) = P((-\infty, \Phi^{-1}(z)]) = [0, z].
    \end{equation}
    Now, let the densities of $Q$ and $P$ be $q$ and $p$, and the densities of $Q'$ and $P'$ be $q'$ and $p'$.
    Then by the change of variables formula
    \begin{equation}
        p'(z) = p\left(\Phi^{-1}(z)\right) (\Phi^{-1})'(z) ~~\text{ and }~~ q'(z) = q\left(\Phi^{-1}(z)\right) (\Phi^{-1})'(z).
    \end{equation}
    Therefore, we have
    \begin{equation}
        \frac{dQ'}{dP'} (z) = \frac{q'(z)}{p'(z)} = \frac{q\left(\Phi^{-1}(z)\right)}{p\left(\Phi^{-1}(z)\right)} = \frac{dQ}{dP} \circ \Phi^{-1}(z),
    \end{equation}
    Now, since $dQ/dP(x)$ is a unimodal function of $x$ and $\Phi^{-1}(z)$ is increasing in $z$, the function $(dQ/dP) \circ \Phi^{-1}(z)$ is unimodal in $z$.
    Also, by taking the the supremum and logarithm of both sizes
    \begin{equation}
        \log \sup_{z \in [0, 1]} \frac{dQ'}{dP'} (z) = \log \sup_{z \in [0, 1]} \frac{dQ}{dP} \circ \Phi^{-1}(z) = \log \sup_{x \in \mathbb{R}} \frac{dQ}{dP}(x),
    \end{equation}
    arriving at the result.
\end{proof}

We now define the super-level sets of the density ratio, and super-level set width functions, on which the argument relies. \\


\begin{definition}[Superlevel set, width]
We define the superlevel-set function $S : [0, 1] \to 2^{[0, 1]}$ as
\begin{equation}
    S(\gamma) = \{x \in [0, 1]~|~ r(x) \geq \gamma \rmax \},
\end{equation}
And let $x_{max} \in \{x \in [0, 1]~|~r(x) \leq r(x_{max})\}$ be an arbitrary maximiser of the density ratio.
We also define the superlevel-set width function $w : [0, 1] \to [0, 1]$ as
\begin{equation}
    w(\gamma) = \inf \{\delta \in [0, 1]~|~\exists z \in [0, 1],~ S(\gamma) \subseteq [z, z + \delta] \}.
\end{equation}
\end{definition}

Because width functions are defined in terms of a ratio of probability densities, they satisfy certain properties, stated in \cref{lemma:propw} and proved below.
We use these properties later to prove \cref{lemma:wcw}. \\


\begin{lemma}[Properties of $w$] \label{lemma:propw}
The width function $w(\gamma)$ is non-increasing in $\gamma$ and satisfies
\begin{equation}
    \int_0^1 w(\gamma)~d\gamma = \frac{1}{\rmax} \text{ and } w(0) \geq \frac{1}{r_{max}}.
\end{equation}
\end{lemma}
\begin{proof}
    First, we note that if $\gamma_1 \leq \gamma_2$, then $S(\gamma_2) \subseteq S(\gamma_1)$ which implies $w(\gamma_2) \leq w(\gamma_1)$.
    Therefore
    \begin{equation}
         \gamma_1 \leq \gamma_2 \implies w(\gamma_2) \leq w(\gamma_1),
    \end{equation}
    so $w(\gamma)$ is decreasing in $\gamma$.
    Second, let $A = [0, 1] \times [0, \rmax]$, define $B = \left\{ (x, y) \in A ~|~ y \leq r(x) \right\}$ and
    consider the integral
    \begin{equation}
        I = \int_A \mathbbm{1}(z \in B)~dz.
    \end{equation}
    Since this the integrand is a non-negative measurable function, by Fubini's theorem, we have
    \begin{align}
        I = \int_0^1 \int_0^{\rmax} \mathbbm{1}((x, y) \in B)~dy~dx &= \int_0^{\rmax} \int_0^1 \mathbbm{1}((x, y) \in B)~dx~dy \\
        \int_0^1 r(x)~dx &= \int_0^{\rmax} w(y/\rmax)~dy \\
        \int_0^1 q(x)~dx &= \int_0^1 w(\gamma) \rmax~d\gamma \\
        \int_0^1 w(\gamma)~d\gamma &= \rmax^{-1}.
    \end{align}
    Last, since $w$ is non-increasing, we have $w(0) \geq \rmax^{-1}$, because otherwise $\int_0^1 w(\gamma)~ d\gamma < \rmax^{-1}$.
\end{proof}

Now we define the two stages in which we break down the execution AS* coding.
In particular, we define $N(\gamma)$ as the number of steps required until AS* gives a sample in the superlevel set $S(\gamma)$, and we define $K(\gamma)$ as the number of subsequent steps required for AS* to terminate. \\


\begin{definition}[\# steps to $S(\gamma)$, \# residual steps]
Suppose AS* is applied to a target-proposal pair $Q, P$ satisfying Assumption \ref{assumption:runtime_proof}, producing a sequence of samples $X_1, X_2, \dots$.
We use $T \in \mathbb{Z}$ to denote the total number of steps taken by AS* until it terminates and define the random variables
\begin{equation}
    N(\gamma) = \min \{n \in \mathbb{Z}~|~X_n \in S(\gamma)\} ~\text{ and }~ K(\gamma) = \max\{0, T - N(\gamma)\}.
\end{equation}
\end{definition}

Because the bounds of AS* shrink exponentially quickly, we can bound the probability that AS* in the first stage, by a quantity which also shrinks exponentially, as stated in \cref{lemma:zw}. \\


\begin{lemma}[Upper bound on the probability of $P(B_n) \geq w(\gamma)$] \label{lemma:zw}
    Let $Z_n = P(B_n)$.
    Then
    \begin{equation}
        \mathbb{P}(Z_n \geq w(\gamma)) \leq \frac{1}{w(\gamma)}\left(\frac{3}{4}\right)^{n-1}
    \end{equation}
\end{lemma}
\begin{proof}
    Let $Z_n = P(B_n)$.
    Noting that $Z_n \geq 0$, we apply Markov's inequality and \cref{lemma:as_star_expected_bound_size_appendix} to get
    \begin{equation}
        \mathbb{P}(Z_n \geq w(\gamma)) \leq \frac{1}{w(\gamma)}\mathbb{E}[Z_n] \leq \frac{1}{w(\gamma)}\left(\frac{3}{4}\right)^{n-1},
    \end{equation}
    as required.
\end{proof}

Now using \cref{lemma:zw} we can upper bound the expectation over $N(\gamma)$, which depends on the logarithm of the width $w(\gamma)$. \\


\begin{lemma}[Bound on expected $N(\gamma)$] \label{lemma:n_bound}
The random variable $N(\gamma)$ satisfies
\begin{equation}
    \mathbb{E}[N(\gamma)] \leq \alpha \log \frac{1}{w(\gamma)} + 6, \text{ where } \alpha =  \left(\log\frac{4}{3} \right)^{-1}.
\end{equation}
\end{lemma}
\begin{proof}
    Let $N_0 = \left\lceil \frac{\log w(\gamma)}{\log (3/4)} \right\rceil$ + 1.
    Also let $B_0, B_1, \dots$ be the bounds produced by AS*.
    Noting that by the unimodality of $r$, $S(\gamma)$ is an interval with $x_{max} \in S(\gamma)$, and $x_{max} \in B_n$, we have
    \begin{equation}
        P(B_n) < w(\gamma) \implies N(\gamma) \leq n,
    \end{equation}
    that is, the event $P(B_n) < w(\gamma)$ implies the event $N(\gamma) \leq n$.
    From this it follows that
    \begin{equation}
        \mathbb{P}(P(B_n) < w(\gamma)) \leq \mathbb{P}(N(\gamma) \leq n) \implies \mathbb{P}(P(B_n) \geq w(\gamma)) \geq \mathbb{P}(N(\gamma) \geq n + 1). \label{eq:b_less_than_w}
    \end{equation}
    Using this together with \cref{lemma:zw}, we can write
    \begin{align}
        \mathbb{E}_{B_{0:\infty}}[N(\gamma)] &= 
        \sum_{n=1}^\infty \mathbb{P}\left(N(\gamma)=n\right)~n \label{eq:ineq:n1} \\
        &= \sum_{n=1}^\infty \mathbb{P}\left(N(\gamma) \geq n\right) \label{eq:ineq:n2} \\
        &\leq N_0 + \sum_{n=N_0 + 1}^\infty \mathbb{P}\left(N(\gamma) \geq n\right) \label{eq:ineq:n3} \\
        &= N_0 + \sum_{n=1}^\infty \mathbb{P}\left(N(\gamma) \geq N_0 + n\right) \label{eq:ineq:n4} \\
        &\leq N_0 + \sum_{n=1}^\infty \mathbb{P}(B_{N_0+n-1} \geq w(\gamma)) \label{eq:ineq:n5} \\
        &\leq N_0 + \sum_{n=1}^\infty \frac{1}{w(\gamma)}\left(\frac{3}{4}\right)^{N_0+n-1} \label{eq:ineq:n6} \\
        &\leq N_0 + \sum_{n=1}^\infty \left(\frac{3}{4}\right)^n \label{eq:ineq:n7} \\
        &= N_0 + 4 \label{eq:ineq:n8} \\
        &\leq \frac{\log w(\gamma)}{\log 3/4} + 6,
    \end{align}
    where the equality of \ref{eq:ineq:n1} and \ref{eq:ineq:n2} is a standard identity \cite{grimmett2014probability}, \ref{eq:ineq:n2} to \ref{eq:ineq:n3} follows by the fact that probabilities are bounded above by 1, \ref{eq:ineq:n3} to \ref{eq:ineq:n4} follows by relabelling the indices, \ref{eq:ineq:n4} to \ref{eq:ineq:n5} follows from \cref{eq:b_less_than_w}, \ref{eq:ineq:n5} to \ref{eq:ineq:n6} follows from follows from \cref{lemma:zw} and \ref{eq:ineq:n5} to \ref{eq:ineq:n6} follows from our definition of $N_0$.
\end{proof}

Now we turn to bounding the expectation of $K(\gamma)$.
For this, we must consider how the difference between the upper and lower bounds maintained by the search shrinks.
To do so, we will use \cref{lemma:mean_neg_g}.
\Cref{lemma:expg} is an intermediate result, which we use to show \cref{lemma:mean_neg_g}. \\


\begin{lemma}[Exponentials and Truncated Gumbels] \label{lemma:expg}
Let $T \sim \text{Exp}(\lambda)$ and $T_0 \geq 0$.
Then
\begin{equation}
    Z \eqdef - \log (T + T_0) \eqdist G ~\text{ where }~ G \sim \text{TG}(\log \lambda,  -  \log T_0).
\end{equation}
\end{lemma}
\begin{proof}
    Let $T \sim \text{Exp}(\lambda)$, $T_0 \geq 0$ and define
    \begin{equation}
        Z \eqdef - \log (T + T_0).
    \end{equation}
    We note that $Z \leq - \log T_0$.
    For $Z \leq - \log T_0$, we can apply the change of variables formula to obtain the density of $Z$.
    Let $p_Z$ and $p_T$ be the densities of $Z$ and $T$. 
    Then
    \begin{align}
        p_Z(z) &= p_T(t) \left| \frac{dt}{dz}\large \right| \\
               &= \lambda e^{-\lambda t} \left| \frac{d}{dz} (e^{-z} - T_0) \right| \\
               &\propto e^{-\lambda e^{-z}} e^{-z} \\
               &\propto e^{-z - e^{-(z - \kappa)}},
    \end{align}
    where $\kappa = \log \lambda$.
    Therefore $Z$ has distribution $\text{TG}(\log \lambda, - \log T_0)$.
\end{proof}


\begin{lemma}[Mean of exponentiated negative truncated Gumbel] \label{lemma:mean_neg_g}
Let $B_1, \dots, B_N$ be the first $N$ bounds produced by AS*, let $G_N$ be the $N^{th}$ Gumbel produced by AS* and define $E_N = e^{-G_N}$.
Then
\begin{equation}
    \mathbb{E}[E_N~|~B_{1:N}] = \sum_{n=1}^N\frac{1}{P(B_n)}.
\end{equation}
\end{lemma}
\begin{proof}
    Define $E_n = e^{-G_n}$ for $n=1, 2, \dots$.
    By the definition of AS*, we have
    \begin{equation}
        G_N ~|~ G_{N-1}, B_{1:N} \sim \text{TG}(\log P(B_N), G_{N-1}).
    \end{equation}
    Negating $G_N$ and $G_{N-1}$, taking exponentials and applying Jensen's inequality together with ineq. (\ref{lemma:expg}), we obtain
    \begin{equation}
        E_N ~|~ T_{N-1}, B_{0:N} \eqdist \tau_N + T_{N-1}, \text{ where } \tau_{N-1} \sim \text{Exp}(P(B_N)).
    \end{equation}
    Repeating this step and taking expectations, we have
    \begin{equation}
        \mathbb{E}[E_N ~|~ B_{0:N}] = \mathbb{E}\left[\sum_{n=1}^N \tau_n ~\big|~ B_{0:N}\right] = \sum_{n=1}^N \frac{1}{P(B_n)}
    \end{equation}
    as required.
\end{proof}

Using \cref{lemma:mean_neg_g} we can bound the expectation over $K(\gamma)$, as stated and proved in \cref{lemma:k_bound}.
Note that this bound does not depend on $N(\gamma)$, which has been marginalised out. \\


\begin{lemma}[Bound on expected $K(\gamma)$] \label{lemma:k_bound}
    The random variable $K(\gamma)$ satisfies
    \begin{equation}
        \mathbb{E}[K(\gamma)] \leq \alpha\left(\log \frac{1}{\gamma} + \log \frac{1}{w(\gamma)}\right) + 16, \text{ where } \alpha = \left(\log \frac{4}{3}\right)^{-1}.
    \end{equation}
\end{lemma}
\begin{proof}
    Let the global upper and lower bounds of AS* at step $n$ be $U_n$ and $L_n$ respectively.
    Then, by the definition of the upper bound of AS* coding
    \begin{equation}
        U_{N(\gamma)} = \log r_{max} + G_{N(\gamma)},
    \end{equation}
    and also, by the definition of the lower bound of AS* coding
    \begin{equation}
        L_{N(\gamma)} = \max_{n \in [1:N(\gamma)]}\big\{\log r(x_n) + G_n \big\} \geq \log r\left(x_{N(\gamma)}\right) + G_{N(\gamma)} \geq \log \gamma r_{max} + G_{N(\gamma)}.
    \end{equation}
    Now for $k = 0, 1, 2, \dots$, we have
    \begin{equation}
        U_{N(\gamma) + k} - L_{N(\gamma) + k} \leq 0 ~\implies~ T \leq N(\gamma) + k,
    \end{equation}
    that is, the event $U_{N(\gamma) + k} - L_{N(\gamma) + k} \leq 0$ implies the event $T \leq N(\gamma) + k$.
    This is because if $U_{N(\gamma) + k} - L_{N(\gamma) + k} \leq 0$, then the algorithm has terminated by step $N(\gamma) + k$, so it follows that $T \leq N(\gamma) + k$.
    Further
    \begin{align}
        U_{N(\gamma) + k} - L_{N(\gamma) + k} &\leq U_{N(\gamma) + k} - L_{N(\gamma)} \\
        &\leq \log r_{max} + G_{N(\gamma)+k} - \log \gamma r_{max} - G_{N(\gamma)} \\
        &= \log \frac{1}{\gamma} + G_{N(\gamma)+k} -  G_{N(\gamma)}.
    \end{align}
    Therefore, we have
    \begin{equation}
        G_{N(\gamma)+k} -  G_{N(\gamma)} \leq \log \gamma \implies T \leq N(\gamma) + k \implies K(\gamma) \leq k,
    \end{equation}
    that is, the event $G_{N(\gamma)+k} -  G_{N(\gamma)} \leq \log \gamma$ implies the event $K(\gamma) \leq k$.
    This holds because if $G_{N(\gamma)+k} -  G_{N(\gamma)} \leq \log \gamma$, then $U_{N(\gamma) + k} - L_{N(\gamma) + k} \leq 0$, which in turn implies $K(\gamma) \leq k$.
    Therefore
    \begin{equation} \label{eq:gk}
        \mathbb{P}\left(G_{N(\gamma)+k} -  G_{N(\gamma)} \leq \log \gamma \right) \leq \mathbb{P}\left(K(\gamma) \leq k \right) \implies
        \mathbb{P}\left(G_{N(\gamma)+k} -  G_{N(\gamma)} \geq \log \gamma \right) \geq \mathbb{P}\left(K(\gamma) \geq k+1\right).
    \end{equation}
    \Cref{eq:gk} upper bounds the probability that the second stage of the algorithm has not terminated, by the probability that the Gumbel values have decreased sufficiently.
    To proceed, we turn to lower bounding the probability of the complementary event $G_{N(\gamma)+k} -  G_{N(\gamma)} \leq \log \gamma$.
    Let $\Phi_{TG}(g; \mu, \kappa)$ denote the CDF of a truncated Gumbel distribution with location parameter $\mu$ and unit scale parameter, truncated at $\kappa$.
    Then
    {
    }
    \begin{align}
        \mathbb{P}\big( &G_{N(\gamma)+n} - G_{N(\gamma)} \leq \log \gamma ~|~ N(\gamma), G_{N(\gamma)}, B_{0:N(\gamma)+n} \big) = \\
        &= \mathbb{E}_{G_{N(\gamma)+n-1}}\left[\mathbb{P}\big( G_{N(\gamma)+n} - G_{N(\gamma)} \leq \log \gamma ~|~N(\gamma), G_{N(\gamma)}, G_{N(\gamma)+n-1}, B_{0:N(\gamma)+n} \big)\right] \\
        &= \mathbb{E}_{G_{N(\gamma)+n-1}}\left[\Phi_{TG}\left(\log \gamma + G_{N(\gamma)};~ \log P(B_{N(\gamma)+n}),~ G_{N(\gamma)+n-1} \right)\right] \\
        &\geq \Phi_{TG}\left(\log \gamma + G_{N(\gamma)};~ \log P(B_{N(\gamma)+n}),~ \infty \right) \\
        &= e^{-e^{- \left(\log \gamma + G_{N(\gamma)} - \log P\left(B_{N(\gamma)+n}\right)\right)}} \\
        &= e^{-\frac{1}{\gamma}~ P\left(B_{N(\gamma)+n}\right)~ e^{-G_{N(\gamma)}}}.
    \end{align}
    Taking an expectation over $G_{N(\gamma)}$ and $ B_{0:N(\gamma)+n}$, we have
    \begin{align}
        \mathbb{P}\big( G_{N(\gamma)+n} - G_{N(\gamma)} \leq \log \gamma ~|~ N(\gamma) \big) &\geq \mathbb{E}_{G_{N(\gamma)}, B_{0:N(\gamma)+n}}\left[ e^{-\frac{1}{\gamma}~ P\left(B_{N(\gamma)+n}\right)~ e^{-G_{N(\gamma)}}} ~\Big|~ N(\gamma) \right] \\
        &\geq e^{-\frac{1}{\gamma}~ \mathbb{E}_{G_{N(\gamma)}, B_{0:N(\gamma)+n}}\left[P\left(B_{N(\gamma)+n}\right)~ e^{-G_{N(\gamma)}}~|~ N(\gamma) \right]} \label{eq:K1}
    \end{align}
    
    Focusing on the term in the exponent, we have
    \begin{align}
        &\mathbb{E}_{G_{N(\gamma)}, B_{0:N(\gamma)+n}}\left[P\left(B_{N(\gamma)+n}\right)~ e^{-G_{N(\gamma)}}~\big|~ N(\gamma) \right] = \\
        &= \mathbb{E}_{B_{0:N(\gamma)-1}}\left[\mathbb{E}_{G_{N(\gamma)}, B_{N(\gamma):N(\gamma)+n}}\left[ P\left(B_{N(\gamma)+n}\right)~ e^{-G_{N(\gamma)}} ~\big|~ B_{0:N(\gamma)-1}, N(\gamma) \right]~\Big|~ N(\gamma) ~\right] \\
        &\leq \mathbb{E}_{B_{0:N(\gamma)-1}}\left[\mathbb{E}_{G_{N(\gamma)}}\left[ \left(\frac{3}{4}\right)^{n+1} P\left(B_{N(\gamma)-1}\right)~ e^{-G_{N(\gamma)}} ~\bigg|~ B_{0:N(\gamma)-1} \right]~\bigg|~ N(\gamma) ~\right] \\
        &= \mathbb{E}_{B_{0:N(\gamma)-1}}\left[ \left(\frac{3}{4}\right)^{n+1} P\left(B_{N(\gamma)-1}\right) \sum_{n=0}^{N(\gamma)-1} \frac{1}{P(B_n)} ~\bigg|~ N(\gamma) ~\right] \\
        &\leq \left(\frac{3}{4}\right)^{n+1} N(\gamma). \label{eq:K2}
    \end{align}
    Substituting \cref{eq:K2} into \cref{eq:K1}, we obtain
    \begin{equation}
        \mathbb{P}\big( G_{N(\gamma)+n} - G_{N(\gamma)} \leq \log \gamma ~|~ N(\gamma) \big) \geq e^{-\frac{N(\gamma)}{\gamma}~ \left(\frac{3}{4}\right)^{n+1}}
    \end{equation}
    and applying ineq. (\ref{eq:ineq:n8}) to this we obtain
    \begin{equation} \label{eq:ineq:gkn}
        \mathbb{P}\big( G_{N(\gamma)+n} - G_{N(\gamma)} \leq \log \gamma \big) \geq e^{-\frac{1}{\gamma}~ \left(\frac{3}{4}\right)^{n+1}~\left(N_0 + 4 \right)},
    \end{equation}
    arriving at a deterministic lower bound on which does not depend on any random quantities.
    Now we also have
    \begin{align}
        \log \frac{1}{\gamma} + \log(N_0 + 4) &= \log \frac{1}{\gamma} + \log\left(\left\lceil \frac{\log w(\gamma)}{\log (3/4)} \right\rceil + 5\right) \\
        &\leq \log \frac{1}{\gamma} + \log\left( \frac{\log w(\gamma)}{\log (3/4)} + 6\right) \label{eq:ineq:loglog} \\
        &\leq \log \frac{1}{\gamma} + \log \frac{1}{w(\gamma)} + 2, \label{eq:ineq:singlelog}
    \end{align}
    where going from \ref{eq:ineq:loglog} to \ref{eq:ineq:singlelog} can be verified numerically.
    Therefore, letting $K_0 = \left\lceil \frac{\log (1/\gamma)~+~\log (1/w(\gamma))~+~2}{\log (4/3)} \right\rceil$, we have
    \begin{align}
        \mathbb{E}\left[K(\gamma)\right] &= \sum_{k=0}^\infty \mathbb{P}\big( K(\gamma) = k \big) ~k \label{eq:ineq:k1} \\
        &= \sum_{k=0}^\infty \mathbb{P}\big( K(\gamma) \geq k \big)  \label{eq:ineq:k2} \\
        &\leq K_0  + \sum_{k=K_0 + 1}^\infty \mathbb{P}\big( K(\gamma) \geq  k \big)  \label{eq:ineq:k3} \\
        &= K_0 + \sum_{k=1}^\infty \mathbb{P}\big( K(\gamma) \geq K_0 + k \big) \label{eq:ineq:k4} \\
        &\leq K_0 + \sum_{k=1}^\infty \mathbb{P}\big( G_{N(\gamma)+K_0+k-1} - G_{N(\gamma)} > \log \gamma \big) \label{eq:ineq:k5} \\
        &\leq K_0 + \sum_{k=1}^\infty \left(1 - e^{-\left(\frac{3}{4}\right)^k}\right) \label{eq:ineq:k6} \\
        &\leq K_0 + 4 \label{eq:ineq:k7} \\
        &\leq \frac{\log \gamma}{\log (3/4)} + \frac{\log w(\gamma)}{\log (3/4)} + 16. \label{eq:ineq:k8}
    \end{align}
    where the equality of \ref{eq:ineq:k1} and \ref{eq:ineq:k2} is a standard identity \cite{grimmett2001probability}, \ref{eq:ineq:k2} to \ref{eq:ineq:k3} follows because probabilities are bounded above by $1$, \ref{eq:ineq:k3} to \ref{eq:ineq:k4} follows by relabelling the indices, \ref{eq:ineq:k4} to \ref{eq:ineq:k5} follows by ineq. (\ref{eq:gk}), \ref{eq:ineq:k5} to \ref{eq:ineq:k6} follows by ineq. (\ref{eq:ineq:gkn}) and the definition of $K_0$, \ref{eq:ineq:k6} to \ref{eq:ineq:k7} can be verified by evaluating the sum using numerical means and \ref{eq:ineq:k7} to \ref{eq:ineq:k8} follows by the definition of $K_0$.
\end{proof}

Putting \cref{lemma:n_bound,lemma:k_bound} together, we obtain \cref{corollary:t_bound}, which is a bound on the expected runtime of AS*.
This holds for any width function $w$ and any $\gamma \in [0, 1]$.
Note that whenever $\gamma = 0$ or $w(\gamma) = 0$ this bound becomes vacuous. \\


\begin{corollary}[Upper bound on $T$ for given $w$] \label{corollary:t_bound}
    For any $\gamma \in [0, 1]$, the total number of steps, $T$, satisfies
    \begin{equation}
        \mathbb{E}[T] \leq 2\alpha \left(\log \frac{1}{w(\gamma)} + 2 \log \frac{1}{\gamma} \right)+ 22, ~\text{ where }~ \alpha = \left(\log \frac{4}{3}\right)^{-1}
    \end{equation}
\end{corollary}
\begin{proof}
    By the definition of $N(\gamma)$ and $K(\gamma)$, we have
    \begin{equation}
        T \leq N(\gamma) + K(\gamma) \implies \mathbb{E}[T] \leq \mathbb{E}\left[N(\gamma) + K(\gamma)\right],
    \end{equation}
    for all $\gamma \in [0, 1]$.
    From \cref{lemma:n_bound} and \cref{lemma:k_bound}, we have
    \begin{equation} \label{eq:looser}
        \mathbb{E}[T] \leq \alpha \left(2\log \frac{1}{w(\gamma)} + \log \frac{1}{\gamma} \right)+ 22 \leq 2\alpha \left(\log \frac{1}{w(\gamma)} + 2\log \frac{1}{\gamma} \right)+ 22,
    \end{equation}
    where $\alpha = \log (4/3)^{-1}$ as required.
\end{proof}

Note that in \cref{eq:looser} we obtain a bound which we intentionally make looser.
This step results in a looser bound but facilitates subsequent manipulations easier.
While a more careful analysis may result in a tighter bound, it can only improve our bound by a scaling factor, and we leave this as a point for further work.

Since \cref{corollary:t_bound} holds for any $\gamma \in [0, 1]$, we can minimise the right hand side with respect to $\gamma$ to make the bound as tight as possible.
This results in a bound that is a function of $w$, however we are interested in producing a bound that holds for all $w$.
Therefore, after minimising with respect to $\gamma$, we will consider the worst possible width functions which maximise the resulting quantity, and show that even for these worst-case width functions, the bound is linear in $\rmax$. 
\Cref{def:fghw} introduces the family of these worst-case width functions for a given $\rmax$. \\


\begin{definition}[Bound functions $f, g, h$, worst-case width set $W^*$] \label{def:fghw}
    We define
    \begin{equation}
        f(\gamma, w) = \log \frac{1}{w(\gamma)},~ g(\gamma) = 2\log \frac{1}{\gamma} ~ \text{ and } ~ h(\gamma, w) = f(\gamma, w) + g(\gamma).
    \end{equation}
    For fixed $\rmax$, let $W(\rmax)$ be the set of all possible width functions.
    We define the set $W^*$ of worst-case width functions as
    \begin{equation}
        W^* = \left\{w^* \in W(\rmax) ~\Big|~ \inf_{\gamma'} h(\gamma', w^*) \geq \inf_{\gamma'} h(\gamma', w) ~\forall~ w \in W(\rmax) \right\}.
    \end{equation}
    We refer to members of this set as worst-case width functions.
\end{definition}

Next, for a given $\rmax$, we define a width function $\tilde{w}$ with a particular form, and show that $\tilde{w} \in W(\rmax)$.
We also show that if $w \in W(\rmax)$ is any other width function, then
\begin{equation}
    \inf_{\gamma} h(\gamma, \tilde{w}) \geq \inf_{\gamma} h(\gamma, w),
\end{equation}
from which it follows that $\tilde{w}$ is a worst case width function, that is $\tilde{w} \in W^*$.


\begin{lemma}[An explicit worst case width function] \label{lemma:wcw}
    The function
    \begin{equation} \label{eq:w_form}
        \tilde{w}(\gamma) = \begin{cases}
            1 & \text{ for } 0 \leq \gamma \leq \tilde{\gamma} \\
            (\tilde{\gamma}/\gamma)^2 & \text{ for } \tilde{\gamma} < \gamma \leq 1
        \end{cases},
    \end{equation}
    where $\tilde{\gamma} = 1 - \sqrt{1 - \rmax^{-1}}$,
    is a width function and $\tilde{w} \in W(\rmax)$.
    Further, if $w \in W(\rmax)$ then
    \begin{equation}
        \inf_{\gamma} h(\gamma, \tilde{w}) \geq \inf_{\gamma} h(\gamma, w).
    \end{equation}
\end{lemma}
\begin{proof}
    Suppose $w \in W(\rmax)$ and let
    \begin{equation}
        m = \inf_{\gamma} h(\gamma, \tilde{w}),
    \end{equation}
    let $\gamma_m$ be the point where $g$ equals $m$, that is
    \begin{equation}
        g(\gamma_m) = 2\log \frac{1}{\gamma_m} = m \implies \gamma_m = e^{-m/2}.
    \end{equation}
    Define $v : [0, 1] \times [0, 1] \to [0, 1]$ as
    \begin{equation}
        v(\gamma, \gamma') = \begin{cases}
            1 & 0 \leq \gamma \leq \gamma' \\
            (\gamma'/\gamma)^2 & \gamma' < \gamma \leq 1
        \end{cases},
    \end{equation}
    and consider $v(\gamma, \gamma_m)$ as a function of $\gamma$.
    Note that $v(\gamma, \gamma_m)$ may not be in $W(\rmax)$ because, while it is non-increasing and continuous, it may not integrate to $\rmax^{-1}$.
    In particular it holds that
    \begin{equation}
        h(\gamma, v(\gamma, \gamma_m)) \leq h(\gamma, \tilde{w}) ~\text{ for all }~ \gamma \in [0, 1] \implies
        v(\gamma, \gamma_m) \geq w(\gamma) \implies \int_0^1 v(\gamma, \gamma_m)~d\gamma \geq \rmax^{-1}.
    \end{equation}
    Now, note that
    \begin{equation}
        \int_0^1 v(\gamma, \gamma')~d\gamma = 2\gamma' - (\gamma')^2.
    \end{equation}
    By the intermediate value theorem, there exists some $0 < \tilde{\gamma} \leq \gamma_m$ such that $2\tilde{\gamma} - \tilde{\gamma}^2 = \rmax^{-1}$.
    For this $\tilde{\gamma}$, we define $\tilde{w}(\gamma) = v(\gamma, \tilde{\gamma})$, which is a width function because it is decreasing and integrates to $1$.
    Specifically, $\tilde{w}(\gamma)$ is in $W(\rmax)$ because the probability density function
    \begin{equation}
        q(x) = \rmax \min\left\{1, \tilde{\gamma} x^{-1/2}\right\},
    \end{equation}
    has $\tilde{w}(\gamma)$ as its width function.
    In addition note that $\tilde{\gamma} \leq \gamma_m$ so we have
    \begin{equation}
        \tilde{\gamma} \leq \gamma_m \implies \inf_\gamma h(\gamma, \tilde{w}(\gamma)) = \inf_\gamma h(\gamma, v(\gamma, \tilde{\gamma})) \geq \inf_\gamma h(\gamma, v(\gamma, \gamma_m)) = \inf_\gamma h(\gamma, w(\gamma)).
    \end{equation}
    Therefore it holds that
    \begin{equation} \label{eq:wtilde}
        w \in W(\rmax) \implies \inf_\gamma h(\gamma, \tilde{w}) \geq \inf_\gamma h(\gamma, w),
    \end{equation}
    from which it follows that $\tilde{w} \in W^*$ is a width function.
\end{proof}

Last, we can put \cref{corollary:t_bound} together with \cref{lemma:wcw} to arrive at the main result. \\

\begin{theorem*}[AS* runtime upper bound]
    Let $T$ be the total number of steps taken by AS* until it terminates.
    Then
    \begin{equation}
        \mathbb{E}[T] \leq 2\alpha \log \rmax + 2\alpha \log 2 + 22.
    \end{equation}
\end{theorem*}
\begin{proof}
    Suppose AS* is applied to a target $Q$ and proposal $P$ with $\infD{Q}{P} = \rmax$, and corresponding width function $w \in W(\rmax)$.
    Now consider the worst case width function $\tilde{w}$ defined in \cref{lemma:wcw}, and note that
    \begin{equation}
        \tilde{\gamma} = 1 - \sqrt{1 - \rmax^{-1}} \geq \frac{1}{2\rmax}.
    \end{equation}
    Then we have
    \begin{align}
        \mathbb{E}[T] &\leq 2\alpha \inf_{\gamma} h(\gamma, w) + 22 \leq 2\alpha \inf_{\gamma} h(\gamma, \tilde{w}) + 22 \leq 2\alpha h(\tilde{\gamma}, \tilde{w}) + 22,
    \end{align}
    and substituting the expression for $h$ we obtain
    \begin{align}
        \mathbb{E}[T] &\leq 2\alpha \left(\log \frac{1}{\tilde{w}(\tilde{\gamma})} + 2\log \frac{1}{\tilde{\gamma}} \right) + 22 \leq 4\alpha \log \rmax + 4\alpha \log 2 + 22,
    \end{align}
    arriving at the result.
\end{proof}

{
}

\section{Proof of \Cref{thm:biasedness_of_a_star}}
\label{section:proof_of_ad_star_bias}
\par
Before we state the precise form of \Cref{thm:biasedness_of_a_star}, we clarify what the precise form the approximate distribution $\widetilde{Q}_D$ of the output of DAD* is for target $Q$ and proposal $P$ for depth limit $D_{max} = D$. This means that there are $N = 2^{D} - 1$ nodes in the binary tree constructed by \cref{alg:a_star_coding} with associated Gumbel values and samples $\{(G_i, X_i)\}_{i = 1}^N$. Let $r = \frac{dQ}{dP}$ and $r_i = r(X_i)$. Then, DAD* searches for $I \defeq \argmax_{i \in [N]}\{ \log r(X_i) + G_i \}$, which can therefore be interpreted as simply performing the Gumbel-max trick on $N$ atoms. Hence, we know, that 
\begin{equation}
\begin{aligned}
    w_i \defeq p(I = i \mid X_{1:N}) &= \frac{\exp(\log r_i)}{\sum_{j = 1}^N \exp(\log r_j) } \\
    &= \frac{r_i}{\sum_{j = 1}^N r_j }.
\end{aligned}
\end{equation}
Hence, we finally get
\begin{equation}
    \widetilde{q}_D(X) = \sum_{i = 1}^N w_i \delta(X - X_i)
\end{equation}
where $\delta$ denotes the Dirac delta function.

Then, the $\widetilde{Q}_D$-expectation of a measurable function $f$ is
\begin{equation}
\label{eq:q_tilde_d_expectation}
    \Exp_{\widetilde{Q}_D}[f] = \sum_{i = 1}^N w_i f(X_i) = \frac{\sum_{i = 1}^N r_i f(X_i)}{\sum_{j = 1}^N r_j}.
\end{equation}
Now, define
\begin{equation}
    \ApproxExp_D(f) \defeq \frac{1}{N} \sum_{i = 1}^N f(X_i)r_i.
\end{equation}
Note, that $\ApproxExp_D(f)$ looks very similar to the usual importance sampling estimator for $Q$-expectation the function $f$ using $N$ samples with distribution $P$.
However, in this case the $X_i$s used in $\ApproxExp_D(f)$ are not identically distributed, though they are independent.
In particular, let $n$ be a node in the tree realized by \cref{alg:priority_queue_construction} with $D_{max} = D$. 
Then, $X_n \sim P\lvert_{B_n}$. The importance of $\ApproxExp_D(f)$ is that the $\widetilde{Q}_D$-expectation of $f$ in \cref{eq:q_tilde_d_expectation} can be written as
\begin{equation}
\label{eq:relation_between_exp_q_tiled_and_imp_est}
    \Exp_{\widetilde{Q}_D}[f] = \frac{\ApproxExp_D(f)}{\ApproxExp_D(1)},
\end{equation}
where $1$ in the above is the constant function identically equal to $1$.
Hence, we begin by investigating the properties of $\ApproxExp_D(f)$. First, we show that it is unbiased:

\begin{lemma}[Unbiasedness of $\ApproxExp_D(f)$]
\label{lemma:unbiasedness_of_importance_estimator}
Let $Q$, $P$, $r$, $f$ and $\ApproxExp_D(f)$ be defined as above. Then,
\begin{equation}
    \Exp_{P(X_{1:N})}[\ApproxExp_D(f)] = \Exp_{Q}[f].
\end{equation}
\end{lemma}
\begin{proof}
Let $\Tree_D$ denote the set of all nodes in the binary tree constructed by \cref{alg:priority_queue_construction} with $D_{max} = D$.
We first note, that for a node $n \in \Tree_D$, we know $X_n \sim P\lvert_{B_n}$, hence $dP(X_n)\lvert_{B_n} = \frac{dP(X_n)}{P(B_n)}$ for $X_n \in B_n$. Furthermore, since \texttt{partition} is dyadic, we know that for a fixed depth $1 \leq d \leq D$, for every node $n$ with $D_n = d$ it holds, that $P(B_n) = 2^{-(d - 1)}$.
Furthermore, by construction, the bounds at a given depth are disjoint and form a partition of the whole space. Therefore,
\begin{equation}
\begin{aligned}
    \Exp_{p(X_{1:N})}&[\ApproxExp_D(f)] = \\ &= \Exp_{p(X_{1:N})}\left[ \frac{1}{N}\sum_{n \in \Tree_D}r_n f(X_n) \right] \\
    &= \Exp_{p(X_{1:N})}\left[ \frac{1}{N}\sum_{d = 1}^D\sum_{\substack{n \in \Tree_D \\ D_n = d}}r_n f(X_n) \right] \\
    &= \frac{1}{N}\sum_{d = 1}^D\sum_{\substack{n \in \Tree_D \\ D_n = d}} \int_{B_n}f(X_n)\frac{dQ}{dP}(X_n)\frac{dP(X_n)}{P(B_n)} \\
    &= \frac{1}{N}\sum_{d = 1}^D 2^{d - 1}\sum_{\substack{n \in \Tree_D \\ D_n = d}} \int_{B_n}f(X_n)dQ(X_n) \\
    &= \frac{1}{N}\sum_{d = 1}^D 2^{d - 1}\int_{\Omega}f(X_n)dQ(X_n) \\
    &= \Exp_{Q}[f]\frac{1}{N}\sum_{d = 1}^D 2^{d - 1} \\
    &= \Exp_{Q}[f],
\end{aligned}
\end{equation}
as required.
\end{proof}

Given the unbiasedness of $\ApproxExp_D(f)$, the rest of the proof follows mostly that of \citet{chatterjee2018sample} with appropriate modifications in the necessary places. Hence, we begin with the following lemma:

\begin{lemma}[Mean absolute deviation bound for $\ApproxExp_D(f)$]
\label{lemma:mean_abs_dev_of_imp_estimator}
Let $Q$, $P$, $r$, $f$ and $\ApproxExp_D(f)$ be defined as above. Let 
\begin{equation}
    K \defeq \lfloor \KLD{Q}{P}_2 \rfloor,~~D \defeq K + t,~~N \defeq 2^D,
\end{equation}
where $t$ is a non-negative integer, and $Y \sim Q$. Define
\begin{equation}
\norm{f}_Q \defeq \sqrt{\Exp_{Q}[f^2]}.
\end{equation}
Let
\begin{equation}
    b \defeq \left(2^{-t/4}\sqrt{1 + \frac{1}{N}} + 2\sqrt{\Prob\left[ \log_2 r(Y) > K + t/2 \right]} \right)
\end{equation}
Then, 
\begin{equation}
    \Exp[\abs{I_{\Tree_P(D)}(f) - \Exp_Q[f]}] \leq 
    b\norm{f}_Q.
\end{equation}
\end{lemma}
\begin{proof}
Let $a = 2^{K + t/2}$, and define
\begin{equation}
    h(x) = f(x)\Ind[r(x) \leq a],
\end{equation}
where $\Ind[\cdot]$ is an indicator function. Let $\phi, \psi \in L^2(Q)$, and define
\begin{equation}
    \innerProd{\phi}{\psi}_Q \defeq \int_{\Omega} \phi \psi dQ.
\end{equation}
By the triangle inequality,
\begin{equation}
\begin{aligned}
    &\abs{I_D(f) - \Exp_Q[f]} \leq \\
    &~~\abs{\Exp_Q[f] - \Exp_Q[h]} + \abs{\ApproxExp_D(f) - \ApproxExp_D(h)} + \abs{I_D(h) - \Exp_Q[h]}.
\end{aligned}
\end{equation}
We will now proceed to bound each term on the right hand side of the inequality. First,
\begin{equation}
\begin{aligned}
    \abs{\Exp_Q[f] - \Exp_Q[h]} &= \abs{\Exp_Q[f - h]} \\
    &\leq \Exp_Q[\abs{f}\Ind[r > a]] \\
    &=\innerProd{\abs{f}}{\Ind[r > a]}_Q \\
    &\leq \norm{f}_Q\sqrt{\Prob[r(Y) > a]},
\end{aligned}
\end{equation}
where the first inequality follows by Jensen's inequality and the second follows by Cauchy-Schwarz. Next,
\begin{equation}
\begin{aligned}
    \Exp[\abs{\ApproxExp_D(f) - \ApproxExp_D(h)}] 
    &= \Exp[\abs{\ApproxExp_D(f - h)}] \\
    &\leq \Exp[\ApproxExp_D(\abs{f - h})] \\
    &= \Exp_Q[\abs{f}\Ind[r > a]] \\
    &=\innerProd{\abs{f}}{\Ind[r > a]}_Q \\
    &\leq \norm{f}_Q\sqrt{\Prob[r(Y) > a]},
\end{aligned}
\end{equation}
where the first inequality holds by Jensen's inequality, the second equality follows by \Cref{lemma:unbiasedness_of_importance_estimator} and the second inequality follows by Cauchy-Schwarz.
Finally,
\begin{equation}
\label{eq:mean_abs_dev_third_term_bound}
\begin{aligned}
    \Exp[\abs{I_D(h) &- \Exp_Q[h]}]^2 \leq \\
    &\leq \Exp[(I_D(h) - \Exp_Q[h])^2] \\
    &= \Var[I_D(h)] \\
    &= \Var\left[ \frac{1}{N}\sum_{d = 1}^D\sum_{\substack{n \in \Tree_D \\ D_n = d}}h(X_n)\frac{dQ}{dP}(X_n) \right] \\
    &= \frac{1}{N^2} \sum_{d = 1}^D\sum_{\substack{n \in \Tree_D \\ D_n = d}} \Var_{P\lvert_{B_n}}\left[h(X_n)\frac{dQ}{dP}(X_n)\right],
\end{aligned}
\end{equation}
where the first inequality holds by Jensen's inequality. Now, examining a single variance term in the above sum:
\begin{equation}
\begin{aligned}
    \Var&_{P\lvert_{B_n}}\left[h(X_n)\frac{dQ}{dP}(X_n)\right] \leq \\
    &\leq \Exp_{P\lvert_{B_n}}\left[h(X_n)^2\frac{dQ}{dP}(X_n)^2\right] \\
    &=\int_{B_n}h^2(X_n)\frac{dQ}{dP}(X_n)^2\frac{dP(X_n)}{P(B_n)} \\
    &=2^{d - 1} \int_{B_n}h^2(X_n)r(X_n) dQ(X_n) \\
    &=2^{d - 1} \int_{B_n}f^2(X_n)\Ind[r(X_n)  \leq a]r(X_n) dQ(X_n) \\
    &=a2^{d - 1} \int_{B_n}f^2(X_n) dQ(X_n).
\end{aligned}
\end{equation}
Plugging this back into \cref{eq:mean_abs_dev_third_term_bound}, we get
\begin{equation}
\begin{aligned}
    \Exp[\abs{I_D(h) &- \Exp_Q[h]}]^2 \leq \\
    &\leq \frac{a}{N^2} 2^{d - 1} \sum_{d = 1}^D\sum_{\substack{n \in \Tree_D \\ D_n = d}} \int_{B_n}f^2(X_n) dQ(X_n) \\
    &= \frac{a}{N^2} \sum_{d = 1}^D 2^{d - 1}  \int_{\Omega}f^2(X_n) dQ(X_n) \\
    &=\frac{a \norm{f}^2_Q}{N^2} \sum_{d = 1}^D 2^{d - 1} \\
    &= \frac{a \norm{f}^2_Q}{N}.
\end{aligned}
\end{equation}
Taking square roots of the very left and very right, we get
\begin{equation}
    \Exp[\abs{I_D(h) - \Exp_Q[h]}] \leq \norm{f}_Q2^{-t/4}\sqrt{1 + \frac{1}{N}}.
\end{equation}
Putting the three bounds together gives us the desired result.
\end{proof}

We are now ready to state our result on the  o $\widetilde{Q}_D$-expectation
\begin{theorem}[Biasedness of DAD* coding]
\label{thm:biasedness_of_a_star_appendix}
Let $Q$, $P$, $r$, $f$, $\ApproxExp$, $\widetilde{Q}_D$, $K$, $D$, $N$ and $Y$ be defined as in \Cref{lemma:mean_abs_dev_of_imp_estimator}. 
Define
\begin{equation*}
    \delta \defeq \left( 2^{-t/4}\sqrt{1 + \frac{1}{N}} + 2\sqrt{\Prob\left[ \log_2 r(Y) > K + t/2 \right]} \right)^{1/2}
\end{equation*}
Then,
\begin{equation}
    \Prob\left[ \abs*{\Exp_{\widetilde{Q}_D}[f] - \Exp_Q[f]} \geq \frac{2\norm{f}_Q \delta}{1 - \delta} \right] \leq 2\delta.
\end{equation}
\end{theorem}
\begin{proof}
The proof is mutatis mutandis the same as the proof of Theorem 1.2 of \citet{chatterjee2018sample}, as it only relies on \cref{eq:relation_between_exp_q_tiled_and_imp_est} and \Cref{lemma:mean_abs_dev_of_imp_estimator}. We simply repeat it here for completeness.
\par
Let $a = 2^{K + t/2}$ and 
\begin{equation}
    b = \left(2^{-t/4}\sqrt{1 + \frac{1}{N}} + 2\sqrt{\Prob\left[ \log_2 r(Y) > K + t/2 \right]} \right).
\end{equation}
Then, by Markov's inequality and \Cref{lemma:mean_abs_dev_of_imp_estimator}, for $\delta \in (0, 1)$ we have
\begin{equation}
\begin{aligned}
    \Prob[\abs{\ApproxExp_D(f) - 1} \geq \delta] &\leq \frac{\Exp[\abs{\ApproxExp_D(f) - 1}]}{\delta} \\
    &\leq \frac{b}{\delta},
\end{aligned}
\end{equation}
and for $\alpha \in (0, 1)$ we get
\begin{equation}
\begin{aligned}
    \Prob[\abs{\ApproxExp_D(f) - \Exp_Q[f]} \geq \alpha] &\leq \frac{\Exp[\abs{\ApproxExp_D(f) - \Exp_Q[f]}]}{\alpha}\\
    &\leq \frac{\norm{f}_Q}{\alpha}.\\
\end{aligned}
\end{equation}
Now, if $\abs{\ApproxExp_D(f) - \Exp_Q[f]} < \alpha$ and $\abs{\ApproxExp_D(f) - 1} < \delta$, then
\begin{equation}
\begin{aligned}
    \abs{\Exp_{\widetilde{Q}_D}[f] &- \Exp_{Q}[f]} \\
    &\leq \abs*{\frac{\ApproxExp_D(f)}{\ApproxExp_D[1]} - \Exp_{Q}[f]} \\
    &\leq \frac{\abs{\ApproxExp_D(f) - \Exp_Q[f] } + \abs{\Exp_Q[f]}\abs{1 - \ApproxExp_D(1)}}{\ApproxExp_D(1)} \\
    &< \frac{\alpha + \abs{\Exp_Q[f]}\delta}{1 - \delta}. 
\end{aligned}
\end{equation}
Finally, setting $\delta = \sqrt{b}$ and $\alpha = \norm{f}_Q\delta$ gives the desired result.

\end{proof}

\section{Proof of \Cref{lemma:ad_star_index_match}}
\label{section:proof_of_ad_star_index_match}
\par

\begin{lemma}
\label{lemma:ad_star_index_match_appendix}
Let $Q$ and $P$ be the target and proposal distributions passed to A* coding (\Cref{alg:a_star_coding}). Let $H^*$ be the heap index returned by unrestricted A* coding and $H^*_d$ be the index returned by its depth-limited version with $D_{max} = d$. Then, conditioned on the public random sequence $S$, we have $H^*_d \leq H^*$. Further, there exists $D \in \Nats$ such that for all $d > D$ we have $H^*_d = H^*$.
\end{lemma}
\begin{proof}
Let $n, m \in \Tree$ be two nodes in the binary tree representation $\Tree$ of the Gumbel process with base measure $P$ realized by \cref{alg:priority_queue_construction} simulated using the public random sequence $S$, such that $D_m < D_n$.
It follows from the definition of heap indexing that $H_m < H_n$. 
Given $S$, A* and its depth-limited version search over the same tree $\Tree$, with the difference that the depth-limited version only searches $\Tree_d \subset \Tree$, the tree truncated after depth $D_{max} = d$. Let 
\begin{equation}
\begin{aligned}
    n^* &= \argmax_{n \in \Tree}\{G_n + \log r(X_n)\} \\
    n^*_d &= \argmax_{n \in \Tree_d}\{G_n + \log r(X_n)\},
\end{aligned}
\end{equation}
the nodes from $\Tree$ returned by unrestricted A* coding and its depth-limited version, respectively. Clearly $D_{n^*_d} \leq d$. Then, we have the following two cases.
\par
\textbf{Case 1:} $d < D_{n^*}$. In this case, we have $D_{n^*_d} < D_{n^*}$, hence $H^*_d < H^*$.
\par
\textbf{Case 2:} $d \geq D_{n^*}$. In this case, depth-limited A* coding finds and returns $n^*$, and since it is the global maximum, increasing the budget further will not make a difference in the returned node. Thus, for this case we get $H^*_d = H^*$, as required. 
\end{proof}

\begin{algorithm}
\SetAlgoLined
\DontPrintSemicolon
\SetKwInOut{Input}{Input}\SetKwInOut{Output}{Output}
\SetKwFunction{pushWithPriority}{pushWithPriority}\SetKwFunction{topPriority}{topPriority}\SetKwFunction{popHighest}{popHighest}\SetKwFunction{partition}{partition}\SetKwFunction{isEmpty}{empty}
\SetKw{yield}{yield}

\Input{Base distribution $P$ with density $p$, number of samples to realize $K$, {\color{blue} maximum search depth $D_{max}$}.}
\Output{Next realization from the Gumbel process $\GumbelProcess_P = (X_k, G_k)_{k = 1}^\infty$.}
$k, D_1, H_1 \gets 1, 1, 1$\;
$Q \gets \mathrm{PriorityQueue}$\;
$G_1 \sim \mathrm{Gumbel}(0)$\;
$X_1 \sim P(\cdot)$\;
$\Pi.\pushWithPriority(1, G_1)$\;

\BlankLine
\While{$!\Pi.\isEmpty()$ }{
    $i \gets \Pi.\popHighest()$\;
    \BlankLine

    \If{{\color{blue}$D_i \leq D_{max}$}}{
        \BlankLine
        $L, R \gets \partition(B_p, X_i)$
        \BlankLine
        \For{$C \in \{L, R\}$}{
            $k \gets k + 1$\;
            $B_k \gets C$\;
            ${\color{blue}D_k \gets D_p + 1}$\;
            ${\color{blue}H_k \gets \begin{cases}
            2H_p & \text{if } C = L \\
            2H_p + 1 & \text{if } C = R \\
            \end{cases}}$\;
            $G_k \sim \mathrm{TruncGumbel}(\log P(B_k), G_p)$\;
            $X_k \sim P(\cdot)\lvert_{B_k}/P(B_k)$\;
            $\Pi.\pushWithPriority(k, G_k)$\;
        }
    }
    \yield{$X_i, G_i, H_i, D_i$}
}
\caption{Priority queue-based top-down construction of a Gumbel process.}
\label{alg:priority_queue_construction}
\end{algorithm}

\section{KL Parameterization for Distributions}
\label{section:iso_kl_layer}
\par
In this section, we demonstrate how Gaussian and uniform target distributions can be parameterized by their KL divergence to some reference distribution, and give some details on how to implement these parameterizations in practice.

\subsection{KL-Mean Parameterization for Gaussians}
\par
Given a reference Gaussian distribution $P = \mathcal{N}(\nu, \rho^2)$, we want to parameterize $Q = \mathcal{N}(\mu, \sigma^2)$ such that
\begin{equation}
    \KLD{Q}{P} = \kappa ~ \text{ and } ~ \sigma < \rho.
\end{equation}
The condition that $\sigma < \rho$ incorporates the additional inductive bias, that since $P$ in practical applications acts as a prior and $Q$ will be representing a variational approximation to a posterior, the posterior should have less uncertainty than the prior. 
It is also a condition required by A* coding, as this condition is necessary and sufficient to ensure $\infD{Q}{P} < \infty$. 
\par
The KL divergence from $Q$ to $P$ is
\begin{equation} \label{eq:kl}
    \KLD{Q}{P} = \log \frac{\rho}{\sigma} + \frac{\sigma^2+  (\mu - \nu)^2}{2\rho^2} - \frac{1}{2} = \kappa.
\end{equation}
We observe that the largest value that $|\mu - \nu|$ can take while satisfying this equality and the constraint $0 < \sigma < \rho$, occurs when $\sigma \to \rho$, in which case we have
\begin{equation}
\label{eq:constrained_mean_diff}
    \frac{(\mu - \nu)^2}{2\rho^2} = \kappa \implies |\mu - \nu| <  \rho ~ \sqrt{2\kappa}.
\end{equation}
Defining $\Delta = (\mu - \nu) / \rho$, we can rearrange \cref{eq:kl} to
\begin{equation}
    \frac{\sigma^2}{\rho^2}~e^{- \sigma^2 / \rho^2} = e^{\Delta^2 - 2\kappa - 1}.
\end{equation}
We can rearrange this equation using the Lambert $W$ function \cite{lambert1758observationes}, into the form
\begin{equation}
    \sigma^2 = -\rho^2 W\left(-e^{\Delta^2 - 2\kappa - 1}\right).
\end{equation}
While the Lambert $W$ does not have an expression in terms of elementary functions, it can be estimated numerically \citep{corless1996lambertw}.
While a numerical solver for the $W$ function is supported in Tensorflow, we found it computationally faster and numerically stabler method to use a Pade approximant of $W$ in practice \cite{Brezinski94extrapolationalgorithms}.
We make this approximatant method available in our code repository.

\par
\textbf{Implementing a Gaussian IsoKL layer:}
Note, that the parameters derived above are in a constrained domain. Thus, given $P = \mathcal{N}(\nu, \rho^2)$, we can reparameterize $\kappa$ and $\mu$ to an unconstrained domain as follows:
\begin{enumerate}
    \item Let $\alpha, \beta$ be real numbers.
    \item Set $\kappa \gets \exp(\alpha)$. This will ensure that $\kappa \geq 0$.
    \item Set $\mu \gets \nu + \rho \sqrt{2\kappa}\tanh(\beta)$. This ensures the inquality on $\abs{\mu - \nu}$ in \Cref{eq:constrained_mean_diff} is satisfied.
    \item Set $\sigma^2 \gets -\rho^2 W\left(-\exp\left(\Delta^2 - 2\kappa - 1\right)\right)$.
\end{enumerate}

\subsection{KL-Infinity Divergence Parameterization for Gaussians}
\label{sec:kl_infd_parameterization}
\par
Assume now, instead of just controlling the KL, we wish to control the R\'enyi $\infty$ divergence as well.
Concretely, for a given reference Gaussian distribution $P = \mathcal{N}(\nu, \rho^2)$, we want to parameterize $Q = \mathcal{N}(\mu, \sigma^2)$ such that
\begin{equation}
    \KLD{Q}{P} = K \quad\text{and}\quad \infD{Q}{P} = R.
\end{equation}
We know, that
\begin{equation}
\begin{aligned}
    K &= \KLD{Q}{P} = \frac{1}{2}\left[ \mu^2 + \sigma^2 - \log \sigma^2 - 1 \right] \\
    R &= \log \sup_{x \in \Reals} \left\{\frac{dQ}{dP}(x) \right\} = \frac{\mu^2}{2(1 - \sigma^2)} - \log \sigma.
\end{aligned}
\end{equation}
From these, we get that
\begin{equation}
\begin{aligned}
    \mu^2 &= 2K - \sigma^2 +\log \sigma^2 + 1\\
    \mu^2 &= 2(1 - \sigma^2) (R + \log \sigma)
\end{aligned}
\end{equation}
Setting these equal to each other
\begin{equation}
\begin{aligned}
    2K - \sigma^2 +\log \sigma^2 + 1 &= 2 R + \log \sigma^2 - 2\sigma^2  R - \sigma^2 \log \sigma^2 \\
    \sigma^2 \log \sigma^2 - \sigma^2 + 2\sigma^2 R &= 2 R - 2K - 1 \\
    \sigma^2 \log \sigma^2 + \sigma^2 (2 R - 1) &= A \\
    \sigma^2 (\log \sigma^2 + B) &= A \\
    \sigma^2 \log (\sigma^2 e^B) &= A \\
    e^B \sigma^2 \log (\sigma^2 e^B) & = Ae^B \\
    e^{\log (\sigma^2 e^B)} \log (\sigma^2 e^B) &=A e^B \\
    \log (\sigma^2 e^B) &= W(Ae^B) \\
    \sigma^2 &= e^{W(Ae^B) - B},
\end{aligned}
\end{equation}
where we made the substitions $A = 2 R - 2K - 1$ and $B = 2R - 1$.
\paragraph{Note:} This parameterization is only unique up to the sign of the target mean $\mu$ due to the symmetry of the Gaussian distribution.

\subsection{KL-Mean parameterization for Uniforms}
\par
Given a uniform reference distribution $P = \Uniform{\nu}{\rho}$ with mean $\nu$ and width $\rho$ on the interval $[\nu - \rho / 2, \nu + \rho / 2]$, we want to parameterize $Q = \Uniform{\mu}{\sigma}$ such that
\begin{equation}
    \KLD{Q}{P} = \kappa.
\end{equation}

Note, that we must have
\begin{equation}
    \nu - \rho / 2 \leq \mu - \sigma / 2 < \mu + \sigma / 2 \leq \nu + \rho / 2
\end{equation}
to ensure $Q \ll P$. Now, we have 
\begin{equation}
    \kappa = \KLD{Q}{P} = \log\frac{\rho}{\sigma},
\end{equation}
from which we get
\begin{equation}
    \sigma = \rho \exp(-\kappa).
\end{equation}

\par
\textbf{Implementing a Uniform IsoKL layer:}
Note, that the parameters derived above are in a constrained domain. Thus, given $P = \Uniform{\nu}{\rho}$, we can reparameterize $\kappa$ and $\mu$ to an unconstrained domain as follows:
\begin{enumerate}
    \item Let $\alpha, \beta$ be real numbers.
    \item Set $\kappa \gets \exp(\alpha)$. This will ensure that $\kappa \geq 0$.
    \item Set $\sigma \gets \rho \exp(-\kappa)$.
    \item Set $\mu \gets \nu + \frac{\rho - \sigma}{2}\tanh(\beta)$. This will ensure $Q \ll P$.
\end{enumerate}


\end{document}

%% file: math_commands.tex
\usepackage{amsmath,amsfonts,bm}

\def\eqref#1{equation~\ref{#1}}

\def\1{\bm{1}}

\DeclareMathAlphabet{\mathsfit}{\encodingdefault}{\sfdefault}{m}{sl}
\SetMathAlphabet{\mathsfit}{bold}{\encodingdefault}{\sfdefault}{bx}{n}

\newcommand{\KL}{D_{\mathrm{KL}}}
\newcommand{\Var}{\mathrm{Var}}

\DeclareMathOperator*{\argmax}{arg\,max}
\DeclareMathOperator*{\argmin}{arg\,min}

%% file: main.bbl
\begin{thebibliography}{37}
\providecommand{\natexlab}[1]{#1}
\providecommand{\url}[1]{\texttt{#1}}
\expandafter\ifx\csname urlstyle\endcsname\relax
  \providecommand{\doi}[1]{doi: #1}\else
  \providecommand{\doi}{doi: \begingroup \urlstyle{rm}\Url}\fi

\bibitem[Agustsson \& Theis(2020)Agustsson and Theis]{agustsson2020universally}
Agustsson, E. and Theis, L.
\newblock Universally quantized neural compression.
\newblock \emph{Advances in Neural Information Processing Systems}, 33, 2020.

\bibitem[Ball{\'e} et~al.(2017)Ball{\'e}, Laparra, and
  Simoncelli]{balle2017end}
Ball{\'e}, J., Laparra, V., and Simoncelli, E.~P.
\newblock End-to-end optimized image compression.
\newblock In \emph{International Conference on Learning Representations}, 2017.

\bibitem[Ball{\'e} et~al.(2018)Ball{\'e}, Minnen, Singh, Hwang, and
  Johnston]{balle2018variational}
Ball{\'e}, J., Minnen, D., Singh, S., Hwang, S.~J., and Johnston, N.
\newblock Variational image compression with a scale hyperprior.
\newblock In \emph{International Conference on Learning Representations}, 2018.

\bibitem[Ball{\'e} et~al.(2020)Ball{\'e}, Chou, Minnen, Singh, Johnston,
  Agustsson, Hwang, and Toderici]{balle2020nonlinear}
Ball{\'e}, J., Chou, P.~A., Minnen, D., Singh, S., Johnston, N., Agustsson, E.,
  Hwang, S.~J., and Toderici, G.
\newblock Nonlinear transform coding.
\newblock \emph{IEEE Journal of Selected Topics in Signal Processing},
  15\penalty0 (2):\penalty0 339--353, 2020.

\bibitem[Bennett et~al.(2002)Bennett, Shor, Smolin, and
  Thapliyal]{bennett2002entanglement}
Bennett, C.~H., Shor, P.~W., Smolin, J.~A., and Thapliyal, A.~V.
\newblock Entanglement-assisted capacity of a quantum channel and the reverse
  {Shannon} theorem.
\newblock \emph{IEEE Transactions on Information Theory}, 48\penalty0
  (10):\penalty0 2637--2655, 2002.

\bibitem[Blundell et~al.(2015)Blundell, Cornebise, Kavukcuoglu, and
  Wierstra]{pmlr-v37-blundell15}
Blundell, C., Cornebise, J., Kavukcuoglu, K., and Wierstra, D.
\newblock Weight uncertainty in neural networks.
\newblock In Bach, F. and Blei, D. (eds.), \emph{Proceedings of the 32nd
  International Conference on Machine Learning}, volume~37 of \emph{Proceedings
  of Machine Learning Research}, pp.\  1613--1622, Lille, France, 07--09 Jul
  2015. PMLR.

\bibitem[Brezinski(1994)]{Brezinski94extrapolationalgorithms}
Brezinski, C.
\newblock Extrapolation algorithms and {Padé} approximations: a historical
  survey, 1994.

\bibitem[Chatterjee \& Diaconis(2018)Chatterjee and
  Diaconis]{chatterjee2018sample}
Chatterjee, S. and Diaconis, P.
\newblock The sample size required in importance sampling.
\newblock \emph{The Annals of Applied Probability}, 28\penalty0 (2):\penalty0
  1099--1135, 2018.

\bibitem[Chewi et~al.(2022)Chewi, Gerber, Lu, Le~Gouic, and
  Rigollet]{chewi2021rejection}
Chewi, S., Gerber, P.~R., Lu, C., Le~Gouic, T., and Rigollet, P.
\newblock Rejection sampling from shape-constrained distributions in sublinear
  time.
\newblock In \emph{International Conference on Artificial Intelligence and
  Statistics}, pp.\  2249--2265. PMLR, 2022.

\bibitem[Corless et~al.(1996)Corless, Gonnet, Hare, Jeffrey, and
  Knuth]{corless1996lambertw}
Corless, R.~M., Gonnet, G.~H., Hare, D.~E., Jeffrey, D.~J., and Knuth, D.~E.
\newblock On the lambertw function.
\newblock \emph{Advances in Computational mathematics}, 5\penalty0
  (1):\penalty0 329--359, 1996.

\bibitem[Dymetman et~al.(2012)Dymetman, Bouchard, and
  Carter]{dymetman2012algorithm}
Dymetman, M., Bouchard, G., and Carter, S.
\newblock The {OS*} algorithm: a joint approach to exact optimization and
  sampling.
\newblock \emph{arXiv preprint arXiv:1207.0742}, 2012.

\bibitem[Flamich et~al.(2020)Flamich, Havasi, and
  Hern{\'a}ndez-Lobato]{flamich2020compressing}
Flamich, G., Havasi, M., and Hern{\'a}ndez-Lobato, J.~M.
\newblock Compressing images by encoding their latent representations with
  relative entropy coding.
\newblock \emph{Advances in Neural Information Processing Systems}, 33, 2020.

\bibitem[Gilks \& Wild(1992)Gilks and Wild]{gilks1992adaptive}
Gilks, W.~R. and Wild, P.
\newblock Adaptive rejection sampling for gibbs sampling.
\newblock \emph{Journal of the Royal Statistical Society: Series C (Applied
  Statistics)}, 41\penalty0 (2):\penalty0 337--348, 1992.

\bibitem[Grimmett \& Stirzaker(2001)Grimmett and
  Stirzaker]{grimmett2001probability}
Grimmett, G. and Stirzaker, D.
\newblock \emph{Probability and random processes}.
\newblock Oxford University Press, U.S.A., 2001.

\bibitem[Grimmett \& Welsh(2014)Grimmett and Welsh]{grimmett2014probability}
Grimmett, G. and Welsh, D.
\newblock \emph{Probability: an introduction}.
\newblock Oxford University Press, 2014.

\bibitem[Harsha et~al.(2007)Harsha, Jain, McAllester, and
  Radhakrishnan]{harsha2007communication}
Harsha, P., Jain, R., McAllester, D., and Radhakrishnan, J.
\newblock The communication complexity of correlation.
\newblock In \emph{Twenty-Second Annual IEEE Conference on Computational
  Complexity (CCC'07)}, pp.\  10--23. IEEE, 2007.

\bibitem[Havasi et~al.(2018)Havasi, Peharz, and
  Hern{\'a}ndez-Lobato]{havasi2018minimal}
Havasi, M., Peharz, R., and Hern{\'a}ndez-Lobato, J.~M.
\newblock Minimal random code learning: Getting bits back from compressed model
  parameters.
\newblock In \emph{International Conference on Learning Representations}, 2018.

\bibitem[Hinton \& Van~Camp(1993)Hinton and Van~Camp]{hinton1993keeping}
Hinton, G.~E. and Van~Camp, D.
\newblock Keeping the neural networks simple by minimizing the description
  length of the weights.
\newblock In \emph{Proceedings of the sixth annual conference on Computational
  learning theory}, pp.\  5--13, 1993.

\bibitem[Ho et~al.(2019)Ho, Lohn, and Abbeel]{ho2019compression}
Ho, J., Lohn, E., and Abbeel, P.
\newblock Compression with flows via local bits-back coding.
\newblock \emph{Advances in Neural Information Processing Systems},
  32:\penalty0 3879--3888, 2019.

\bibitem[Hoogeboom et~al.(2019)Hoogeboom, Peters, van~den Berg, and
  Welling]{hoogeboom2019integer}
Hoogeboom, E., Peters, J., van~den Berg, R., and Welling, M.
\newblock Integer discrete flows and lossless compression.
\newblock \emph{Advances in Neural Information Processing Systems},
  32:\penalty0 12134--12144, 2019.

\bibitem[Kingma \& Welling(2014)Kingma and Welling]{kingma2013auto}
Kingma, D.~P. and Welling, M.
\newblock Auto-encoding variational {Bayes}.
\newblock \emph{International Conference on Learning Representations}, 2014.

\bibitem[Kingman(1992)]{kingman1992poisson}
Kingman, J.
\newblock \emph{Poisson Processes}.
\newblock Oxford Studies in Probability. Clarendon Press, 1992.
\newblock ISBN 9780191591242.

\bibitem[Lambert(1758)]{lambert1758observationes}
Lambert, Johann, H.
\newblock Observationes variae in mathesin puram.
\newblock \emph{Acta Helvetica Physico-Mathematico-Anatomico-Botanico-Medica},
  3:\penalty0 128--168, 1758.

\bibitem[Land \& Doig(1960)Land and Doig]{land1960automatic}
Land, A. and Doig, A.
\newblock An automatic method of solving discrete programming problems.
\newblock \emph{Econometrica}, 28\penalty0 (3):\penalty0 497--520, 1960.

\bibitem[Li \& El~Gamal(2018)Li and El~Gamal]{li2018strong}
Li, C.~T. and El~Gamal, A.
\newblock Strong functional representation lemma and applications to coding
  theorems.
\newblock \emph{IEEE Transactions on Information Theory}, 64\penalty0
  (11):\penalty0 6967--6978, 2018.

\bibitem[Maddison(2016)]{maddison2016poisson}
Maddison, C.
\newblock Poisson process model for {Monte Carlo}.
\newblock \emph{Perturbation, Optimization, and Statistics}, pp.\  193--232,
  2016.

\bibitem[Maddison et~al.(2014)Maddison, Tarlow, and
  Minka]{maddison2014sampling}
Maddison, C.~J., Tarlow, D., and Minka, T.
\newblock A* sampling.
\newblock \emph{Advances in Neural Information Processing Systems},
  27:\penalty0 3086--3094, 2014.

\bibitem[Markou(2022)]{markou2022notes}
Markou, S.
\newblock Notes on the runtime of {A*} sampling.
\newblock \emph{arXiv preprint arXiv.2205.15250}, 2022.

\bibitem[Papandreou \& Yuille(2011)Papandreou and
  Yuille]{papandreou2011perturb}
Papandreou, G. and Yuille, A.~L.
\newblock {Perturb-and-MAP} random fields: Using discrete optimization to learn
  and sample from energy models.
\newblock In \emph{2011 International Conference on Computer Vision}, pp.\
  193--200. IEEE, 2011.

\bibitem[P{\'e}rez-Cruz(2008)]{perez2008kullback}
P{\'e}rez-Cruz, F.
\newblock {Kullback-Leibler} divergence estimation of continuous distributions.
\newblock In \emph{2008 IEEE international symposium on information theory},
  pp.\  1666--1670. IEEE, 2008.

\bibitem[Theis \& Agustsson(2021)Theis and Agustsson]{Theis2021a}
Theis, L. and Agustsson, E.
\newblock On the advantages of stochastic encoders.
\newblock In \emph{Neural Compression Workshop at ICLR}, 2021.

\bibitem[Theis \& Yosri(2022)Theis and Yosri]{theis2021algorithms}
Theis, L. and Yosri, N.
\newblock Algorithms for the communication of samples.
\newblock In \emph{International Conference on Machine Learning}, 2022.

\bibitem[Townsend et~al.(2018)Townsend, Bird, and
  Barber]{townsend2018practical}
Townsend, J., Bird, T., and Barber, D.
\newblock Practical lossless compression with latent variables using bits back
  coding.
\newblock In \emph{International Conference on Learning Representations}, 2018.

\bibitem[Townsend et~al.(2019)Townsend, Bird, Kunze, and
  Barber]{townsend2019hilloc}
Townsend, J., Bird, T., Kunze, J., and Barber, D.
\newblock Hilloc: lossless image compression with hierarchical latent variable
  models.
\newblock In \emph{International Conference on Learning Representations}, 2019.

\bibitem[van~den Berg et~al.(2020)van~den Berg, Gritsenko, Dehghani,
  S{\o}nderby, and Salimans]{van2020idf++}
van~den Berg, R., Gritsenko, A.~A., Dehghani, M., S{\o}nderby, C.~K., and
  Salimans, T.
\newblock {IDF++}: Analyzing and improving integer discrete flows for lossless
  compression.
\newblock In \emph{International Conference on Learning Representations}, 2020.

\bibitem[Zhang et~al.(2021)Zhang, Kang, Ryder, and Li]{zhang2021iflow}
Zhang, S., Kang, N., Ryder, T., and Li, Z.
\newblock iflow: Numerically invertible flows for efficient lossless
  compression via a uniform coder.
\newblock \emph{Advances in Neural Information Processing Systems}, 34, 2021.

\bibitem[Ziv(1985)]{ziv1985universal}
Ziv, J.
\newblock On universal quantization.
\newblock \emph{IEEE Transactions on Information Theory}, 31\penalty0
  (3):\penalty0 344--347, 1985.

\end{thebibliography}
